%% file: 1ea-gct-tensor.tex
\documentclass[11pt]{amsart}
\usepackage{latexsym,amssymb,amsmath,amscd}
\usepackage[enableskew]{youngtab}

\input macros

\title[Geometric Complexity Theory and Tensor Rank]
{\bf Geometric Complexity Theory and Tensor Rank\\[2ex]
  {\rm (Extended Abstract)}}

\author{Peter B\"urgisser and Christian Ikenmeyer}
\address{Institute of Mathematics, University of Paderborn, D-33098 Paderborn, Germany}
\thanks{
Institute of Mathematics, University of Paderborn, D-33098 Paderborn, Germany.\\ 
pbuerg@upb.de, ciken@math.upb.de\\
Supported by DFG-grant BU 1371/3-1.}
\email{pbuerg@upb.de, ciken@math.upb.de}

\keywords{geometric complexity theory, tensor rank, matrix multiplication, orbit closures,
multiplicities, Kronecker coefficients}
\subjclass[2000]{68Q17, 14L24, 20G05}

\date{\today}

\begin{document}
\maketitle

\begin{abstract}
Mulmuley and Sohoni~\cite{gct1,gct2} proposed to view the permanent versus determinant problem
as a specific orbit closure problem and to attack it by methods from geometric invariant
and representation theory.
We adopt these ideas towards the goal of showing lower bounds on the border rank of specific tensors,
in particular for matrix multiplication. We thus study specific orbit closure problems for
the group $G =\GL(W_1)\times\GL(W_2)\times\GL(W_3)$ acting on the tensor product
$W=W_1\otimes W_2\otimes W_3$ of complex finite dimensional vector spaces.
Let $G_s =\SL(W_1)\times\SL(W_2)\times\SL(W_3)$.
A key idea from~\cite{gct2} is that the irreducible $G_s$-representations occurring
in the coordinate ring of the $G$-orbit closure of a stable tensor~$w\in W$
are exactly those having a nonzero invariant with respect to
the stabilizer group of~$w$.

However, we prove that by considering $G_s$-representations,
as suggested in~\cite{gct1,gct2},
only trivial lower bounds on border rank can be shown.
It is thus necessary to study $G$-representa\-tions,
which leads to geometric extension problems that are beyond the scope of
the subgroup restriction problems emphasized in~\cite{gct1,gct2}.
We prove a very modest lower bound on the border rank
of matrix multiplication tensors using $G$-representations.
This shows at least that the
barrier for $G_s$-representations can be overcome. 
To advance, we suggest the coarser approach to replace the semigroup of representations
of a tensor by its moment polytope.
We prove first results towards determining the moment polytopes of
matrix multiplication and unit tensors.
\end{abstract}

\maketitle

\subsection*{Acknowledgments}

We thank
Matthias Christandl, Shrawan Kumar, Joseph Landsberg, Laurent Manivel, Ketan Mulmuley,
and  Jerzy Weyman for helpful discussions.

\newpage
\setcounter{page}{1}
\section{Introduction}

Mulmuley and Sohoni~\cite{gct1,gct2} proposed to view the permanent versus determinant problem
as a specific orbit closure problem and to attack it by methods from geometric invariant
and representation theory.
So far there has been little progress with this approach, mainly due to the difficulty
of the various arising mathematical problems~\cite{blmw:09}.
It is the goal of this paper to examine and to further develop the collection of ideas
from~\cite{gct1,gct2} at a problem simpler than the permanent versus determinant, but still of
considerable interest for complexity theory.

The complexity of matrix multiplication is captured by the rank of the matrix multiplication
tensor, a quantity that, despite intense research efforts, is little understood.
Strassen~\cite{stra:87} already observed that the closely related notion of border rank
has a natural formulation as a specific orbit closure problem.
Moreover, it is remarkable that the best known lower bound on the rank of matrix multiplication
(Bl\"aser~\cite{blas:99}) owes its existence to an explicit construction of an invariant polynomial
in the vanishing ideal of certain secant varieties (Strassen~\cite{stra:83}).

We carried out the program in~\cite{gct1,gct2} for the matrix multiplication versus unit tensor problem.
More specifically, we determined the stabilizers (symmetry groups) of the corresponding tensors
and verified that they are stable. Moreover, we found explicit representation theoretic characterizations
of the irreducible $G_s$-representations occurring in the coordinate rings of the $G$-orbit closures
of these tensors
in terms of nonvanishing of Kronecker coefficients and related quantities
($G$ and $G_s$ stand for a product of general linear groups and special linear groups, respectively,
cf.~\eqref{eq:defG} and \S\ref{se:stab}).

Unfortunately, it turns out that using $G_s$-representations, only trivial lower bounds on border
rank can be shown (Theorem~\ref{th:no-lower-bound})!
This insight is one of our main results.
It does not kill the overall program, but implies, unlike proposed in~\cite{gct1,gct2},
that the finer $G$-representations have to be considered instead.
As a consequence, the stability property is not enough to overcome the issue of orbit closures
and additional properties, beyond the subgroup restriction problems emphasized
in~\cite{gct1,gct2}, need to be studied.
What we have to face is the problem of extending (highest weight) regular functions from 
an orbit to its orbit closure. 
It turns out that this can be captured by a single integer~$k$ that seems of a 
geometric nature (cf.~Theorem~\ref{th:extend}). 
Currently we understand the extension problem very little.

In~\S\ref{se:comp} we prove, for the first time, a lower bound on the border rank
of matrix multiplication tensors using $G$-representations.
While this bound is still very modest, it shows at least that the
barrier for $G_s$-representations from Theorem~\ref{th:no-lower-bound}
can be overcome.

A natural approach to advance is to take a coarser, asymptotic viewpoint which replaces
the semigroup of representations by the their moment polytopes~\cite{brio:87,stra:05}.
We prove first results towards determining the moment polytopes of matrix multiplication and
unit tensors. This is based on the asymptotic properties of the Kronecker polytope
derived in~\cite{buci:09}.

Due to lack of space most of the proofs had to be omitted in this extended abstract.
Full proofs are provided in the appendix.

\section{Preliminaries}\label{se:prelim}

\subsection{Tensor rank}

Let $W_1,W_2,W_3$ be finite dimensional complex vector spa\-ces
of dimensions $m_1,m_2,m_3$, respectively.
We put $W:=W_1\ot W_2\ot W_3$ and call $\um=(m_1,m_2,m_3)$ the
{\em format} of $W$.
The elements $w\in W$ shall be called {\em tensors} and $w$ is
called indecomposable if it has the form $w=w_1\ot w_2\ot w_3$.
The {\em rank} $R(w)$ of $w\in W$ is defined as the minimum~$r\in\N$
such that $w$~can be written as a sum of $r$ indecomposable tensors.
We note that if $W_3=\C$, then $R(w)$ is just the rank of the linear
map $W_1^*\to W_2$ corresponding to~$w$.

Strassen proved~\cite{stra:73} that the minimum number of
nonscalar multiplications sufficient to evaluate the bilinear map
$W_1^*\ti W_2^*\to W_3$ corresponding to $w$
differs from $R(w)$ by at most a factor of two.
Determining the rank of specific tensors turns out to be
very difficult. Of particular interest are the tensors
$\mamun\in (\C^{n\times n})^*\ot(\C^{n \times n})^*\ot\C^{n\times n}$
describing the multiplication of two $n$ by $n$ matrices.
The best known asymptotic upper bound~\cite{cowi:90} states
$R(\mamun) = \Oh(n^{2.38})$,  while
the best known lower bound~\cite{blas:99} is
$R(\mamun) \ge \frac{5}{2}\, n^2 - 3n$.

The {\em border rank} $\br(w)$  of $w\in W$ is defined as the
smallest $r\in \N$ such that $w$ can be obtained as the limit of a sequence
$w_k\in W$ with $R(w_k)\le r$ for all $k$.
Clearly, $\br(w)\le R(w)$.
Border rank is a natural mathematical notion closely related to the rank
and it has played an important role in the discovery of fast algorithms for
matrix multiplication, see~\cite{ACT}.
We note that the best known lower bound on the border rank of matrix multiplication~\cite{lick:84}
states that $\br(\mamun) \ge \frac32\,n^2 + \frac12\, n -1$.

\subsection{Orbit closure problem}\label{se:ocp}

It is possible to rephrase the determination of $\br(w)$ as an orbit closure problem.
Consider the algebraic group
\begin{equation}\label{eq:defG}
G:=\GL(W_1)\ti\GL(W_2)\ti\GL(W_3)
\end{equation}
acting linearly on the vector space $W=W_1\ot W_2\ot W_3$ via
$(g_1,g_2,g_3)(w_1\ot w_2\ot w_3) := g_1(w_1)\ot g_2(w_2)\ot g_3(w_3)$.
We shall denote by $Gw$ the {\em orbit} of~$w$ and by
$\ol{Gw}$ its {\em orbit closure}.
We say that~$w$ is a {\em degeneration}
of~$v$, written $w\unlhd v$,
iff $\ol{Gw}\subseteq\ol{Gv}$.

Suppose now $m\le m_i$ and choose bases $e^i_1,\ldots,e^i_{m_i}$
in each of the spaces~$W_i$.
The tensor
\begin{equation}\label{eq:rthunit}
 \lan m\ran := \sum_{j=1}^r e^1_j\ot e^2_j\ot e^3_j ,
\end{equation}
is called an {\em $m$th unit tensor} in $W$.
Another choice of bases leads to a tensor in the same $G$-orbit as $\lan m\ran$,
so that the orbit of $m$th unit tensors in~$W$ is a basis independent notion.
It is easy to see that
$\br(w)\le m$ iff $w \unlhd \lan m\ran$, cf.~\cite{stra:87}.

\section{The GCT program for tensors}\label{se:GCTprog}

We summarize here in a concise way the stepping stones of the
GCT program ~\cite{gct1,gct2}, adapted to the tensor setting.
For this the review of the GCT program in~\cite{blmw:09} has been very helpful.

\subsection{Semigroups of representations}

For background on representation theory see~\cite{fuha:91,good-wall:09}.
We denote by $V_{\la_i}(\GL(W_i))$ the Schur-Weyl module labelled by
its highest weight $ \la_i\in\Z^{m_i}$
(with monotonically decreasing entries).
Those yield the rational irreducible $G$-modules
$$
 V_\ula(G):=
  V_{\la_1}(\GL(W_1))\ot V_{\la_2}(\GL(W_2))\ot V_{\la_3}(\GL(W_3)) ,
$$
whose highest weights $\ula$ are triples $\ula=(\la_1,\la_2,\la_3)$.
We denote by $V_\ula(G)^* = V_{\ula^*}(G)$ the module dual to~$V_\ula(G)$.
Moreover, $\Lambda^+_G$ shall denote the semigroup of highest weights of~$G$.
For a dimension format $\um$ we consider the subsemigroup
$\Lambda^+(\um):= \bigcup_{d\in\N} \Lambda^+_d(\um)$
of $\Lambda^+_G$, where
$$
 \Lambda^+_d(\um) := \big\{ \ula=(\la_1,\la_2,\la_3) \mid \la_i \vdash_{m_i} d,\ i=1,2,3 \big\} .
$$
Here we use the notation $\la_i\vdash_{m_i} d$ for a partition
$\la_i=(\la_i^1,\ldots,\la_i^{m_i})$ of $d$ into at most $m_i$ parts.

The action of $G$ on $W$ induces a linear action of $G$ on the ring $\Oh(W)$
of polynomial functions on~$W$ via
$(g f)(w) := f(g^{-1}w)$ for $g\in G$, $f\in\Oh(W)$, $w\in W$.
For any tensor $w\in W$, this defines a linear action of $G$
on the graded ring $\Oh(\ol{Gw})=\oplus_{d\in\N} \Oh(\ol{Gw})_d$
of regular functions on $\ol{Gw}$.
(By a regular function on $\overline{Gw}$
we understand a restriction of a polynomial function.)
Since $G$ is reductive, the $G$-module $\Oh(\ol{Gw})_d$ splits
into irreducible $G$-modules.

We define now the main objects of our investigations.

\begin{definition}\label{def:Ssemigroup}
The {\em semigroup of representations}~$S(w)$
of a tensor $w\in W$ is defined as
$$
S(w) := \big\{ \ula \mid \mbox{$V_\ula(G)^*$ occurs in $\Oh(\ol{Gw})$} \big\} .
$$
\end{definition}

It is known that $S(w)$ is a finitely generated subsemigroup of $\Lambda^+(\um)$,
cf.~\cite{brio:87}.
It is easy to see that if
$V_\ula(G)^*$ occurs in degree~$d$, i.e.,
as a submodule of $\Oh(\ol{Gw})_d$,
then $\ula\in \Lambda^+_d(\um)$.

The general strategy of geometric complexity theory~\cite{gct1}
is easily described.
Schur's lemma implies that for $w,v\in W$
\begin{equation}\label{eq:strat}
 \ol{Gw}\subseteq\ol{Gv} \Longrightarrow S(w)\subseteq S(v).
\end{equation}
In particular, exhibiting some $\ula\in S(w)\setminus S(v)$
proves that $\ol{Gw}$ is not contained in $\ol{G v}$.
If $v=\langle m\rangle$, this establishes the lower bound $\br (w) >m$.
We call such $\ula$ a {\em representation theoretic obstruction}.
We note that a more refined approach would be to study the multiplicity of
$V_\ula(G)^*$ in $\Oh(\ol{Gw})$, which can only decrease under degenerations.

\subsection{Kronecker semigroup}\label{se:kron}

Let $[\la_i]$ denote the irreducible representation of the symmetric group~$S_d$
on $d$ letters labelled by the partition $\la_i\vdash d$.
For $\ula\in\Lambda^+_d(\um)$ we define the {\em Kronecker coefficient} $g(\ula)$
as the dimension of the space of $S_d$-invariants
of the tensor product $[\la_1]\otimes  [\la_2]\otimes [\la_3]$.
It is a well known fact that
$g(\ula)  =\mult(V_\ula(G)^*,\Oh(W)_d)$, see~\cite{lama:04} or \eqref{eq:multOW}. 
The {\em Kronecker semigroup} of format~$\um$  is defined by
$K(\um) :=\bigcup_{d\in\N} \{ \ula\in \Lambda^+_d(\um) \mid g(\ula)\ne 0 \big\}$.

\begin{lemma}\label{pro:Sgen}
We have $S(w)\subseteq K(\um)$ with
equality holding for Zariski almost all $w\in W$.
\end{lemma}

\subsection{Inheritance}

For applying the criterion
``$\br(w)\le r$ iff $w\unlhd \langle r\rangle$'' from~\S\ref{se:ocp},
we need to understand how $S(w)$ changes
when we embed $w\in W$ in a larger space. Fortunately, when properly
interpreted, nothing happens´.
Suppose that $W_i$ is a subspace of~$W'_i$,
put $m'_i := \dim W'_i$, and let
$$
 W':= W'_1\ot W'_2\ot W'_3,\quad
 G' := \GL(W'_1)\ti\GL(W'_2)\ti\GL(W'_3).
$$
Let $w'$ denote the image of $w\in W$ under the embedding
$W\hookrightarrow W'$.
A highest $G$-weight $\ula$ with nonnegative entries
can be interpreted as a highest $G'$-weight $\ula$
by appending zeros to the partitions $\la_i$.
We may thus interpret
$\Lambda^+(\um)$ as a subset of $\Lambda^+(\um')$.

\begin{proposition}\label{th:inheritance}
With the above conventions we have $S(w)=S(w')$.
\end{proposition}

This result can be shown similarly as in~\cite{gct2,lama:04,blmw:09}.

\subsection{Stabilizer and invariants}

As a first approach towards understanding $S(w)$
we may replace the orbit closure $\overline{Gw}$ by the
orbit $Gw$
and focus on the representations occuring in the ring
$\Oh(Gw)$ of regular functions on~$Gw$.
(A regular function on $Gw$ is a function that
is locally rational, cf.~\cite[p.~15]{hart:77}.)
This leads to definition
of the {\em auxiliary semigroup of representations:}
$$
 S^o(w) := \big\{ \ula \in \Lambda^+_G \mid
  \mbox{$V_\ula(G)^*$ occurs in $\Oh(Gw)$} \big\} .
$$
$S^o(w)$ is finitely generated~\cite{brio:87} 
and clearly contains $S(w)$.

The {\em stabilizer} of $w$ is defined as
$H:=\stab(w):=\{g\in G \mid gw =w \}$.
Let $V_\ula(G)^{H}$ denote the space of $H$-invariants in $V_\ula(G)$.
The next characterization follows from the algebraic Peter-Weyl Theorem for~$G$
as in~\cite{blmw:09}.

\begin{proposition}\label{pro:pewe}
We have 
$S^o(w) =\big\{ \ula\in\Lambda^+_G \mid V_\ula(G)^H \ne 0 \big\}$.
\end{proposition}

\subsection{Stability}\label{se:stab}

Consider the subgroup $G_s :=\SL(W_1)\times\SL(W_2)\times\SL(W_3)$ of~$G$.

\begin{definition}\label{def:stable}
A tensor $w\in W$ is called {\em stable}, iff $G_sw$ is closed.
\end{definition}

Consider the residue class map
$\prod_{i=1}^3 \Z^{m_i} \to \prod_{i=1}^3 \Z^{m_i}/\Z \e_{m_i}$,
where $\e_{m_i}:=(1,\ldots,1)$.
When interpreting highest weights of $G_s$-modules appropriately,
this defines a surjective morphism
$\pi\colon\Lambda^+_G\to\Lambda^+_{G_s}$
of the semigroup of highest weights of $G$ and $G_s$, respectively.

We put $S_s(w):=\pi(S(w))$ and $S_s^o(w) := \pi(S^o(w))$.
These semigroups describe the irreducible $G_s$-modules occurring in $\Oh(\ol{Gw})$
and $\Oh(Gw)$, respectively.
However, when going over to $G_s$-modules, the information about the degree~$d$
in which the modules occur is lost.


\begin{proposition}\label{th:stabSs}
If $w$ is stable, then $S_s(w)=S^o_s(w)$.
\end{proposition}

\begin{proof}
Put $\e_\um  := (\e_{m_1},\e_{m_2},\e_{m_3})$.
The assertion is equivalent to the statement
\begin{equation}\label{eq:SS0}
\forall \ula \in S^o(w)\ \exists k\in\Z\quad \ula + k\e_\um \in S(w) .
\end{equation}
Suppose that $\ula \in S^o(w)$. Then
$V_\ula(G)^*$ occurs in $\Oh(Gw)_d$ for some $d\in\Z$.
Let $f\in \Oh(Gw)_d$ be a highest weight vector of $V_\ula(G)^*$.
The restriction $\tilde{f}$ of~$f$ to $G_sw$ does not vanish since
$Gw$ is the cone generated by $G_sw$.
So $\tilde{f}$ is a highest weight vector and
$V_{\pi(\ula)}(G_s)^*$ occurs in $\Oh(G_s w)$.

The $G_s$-equivariant restriction morphism
$\Oh(\overline{Gw})\to \Oh(G_s w)$
is surjective since $G_sw$ is assumed to be closed.
It follows that $\Oh(\overline{Gw})$ contains
an irreducible module $V_{\pi(\ula)}(G_s)^*$.
This means that $V_{\ula + k\e_\um}(G)^*$
occurs in $\Oh(\overline{Gw})$ for some $k\in\Z$.
\end{proof}

Combining this with Proposition~\ref{pro:pewe},
we obtain a characterization of $S_s(w)$ for stable tensors~$w$,
which only involves the stabilizer~$H$ of $w$.
The problem is reduced to the question of which
$V_\ula(G)$ contain nonzero $H$-invariants.

\section{Unit tensors}\label{se:unit}

\subsection{Stabilizer and stability}

Suppose $W_i=\C^m$, $G_m:=\GL_m\ti\GL_m\ti\GL_m$,
and recall the definition of the $m$th unit tensor $\lan m\ran$
from~\eqref{eq:rthunit}.
Let $P_\pi\in\GL_m$ denote the permutation matrix corresponding to $\pi\in S_m$.

\begin{proposition}\label{pro:stabut}
The stabilizer $H_m$ of $\lan m\ran$ is the semidirect product of the normal divisor
$$
 \Tc_m:=\{(\diag(a),\diag(b),\diag(c)) \mid \forall i\ a_ib_ic_i=1\},
$$
and the symmetric group $S_m$ diagonally embedded in $G_m$ via
$\pi\mapsto (P_\pi,P_\pi,P_\pi)$.
\end{proposition}

We remark that $\lan m \ran$ is uniquely determined by its stabilizer up to a scalar.

\begin{proposition}\label{th:unit-stable}
The unit tensor $\lan m\ran$ is stable.
\end{proposition}

This follows from Kempf's refinement~\cite{kempf:78}
of the Hilbert-Mumford criterion.

\subsection{Representations}\label{se:rep-unit}

Let $\Par_m(d)$ denote the set of partitions of $d$ into at most $m$ parts.
The {\em dominance order} $\preceq$ on $\Par_m(d)$ is defined by
$\la \preceq \mu$ iff $\sum_{j=1}^k \la^j \le \sum_{j=1}^k \mu^j$ for all~$k$.
This defines a lattice, in particular two partitions $\la, \mu$
have a well defined meet $\la\meet\mu$, cf.~\cite{StanleyVol2}.
We call $\a\in\Par_m(d)$ {\em regular}
if its components are pairwise distinct.

\begin{lemma}\label{le:unique-min}
(1) The set of regular partitions in $\Par_m(d)$ has a unique
smallest element~$\bot_m(d)$.
(2) For any $\la\in\Par_m(d)$ we have
$\bot_{m+1}(d +km) \preceq (\la^1 + k,\ldots,\la^m +k, 0)$
for sufficiently large~$k$.
\end{lemma}

Let $T_m$ denote the maximal torus of $\GL_m$
of diagonal matrices.
For $\a\in\Z^m$ with $|\a| := \sum_j \a_j =d$
and $\la\in\Par_m(d)$ one defines
the  {\em weight space} of $V_\la = V_\la(\GL_m)$ for the
weight $\a$ as
$$
 V_\la^\a := \big\{ v\in V_\la \mid \forall t \in T_m\
  t \cdot v = t^\a \big\} .
$$
Here we used the shorthand notation
$t=\diag(t_1,\ldots,t_m)$ and
$t^\a = t_1^{\a_1}\cdots t_m^{\a_m}$.
It is well known that $V_\la$ decomposes as
$V_{\la} = \bigoplus_{\a} V_{\la}^{\a}$.
Moreover, $V_{\la}^{\a}$ is nonzero iff
$\a\preceq \la$, cf.~\cite{fuha:91}.

The symmetric group $S_m$ acts on $\Z^m$
by permutation, namely $(\pi \a)(i):=\a(\pi^{-1}(i))$
for $\pi\in S_m$.
It is easy to check that
$P_\pi V_\la^\a = V_\la^{\pi\a}$.
In particular, the stabilizer $\stab(\a)$ of~$\a$
leaves $V_\la^\a$ invariant.
Note that $\stab(\a)$ is trivial iff $\a$ is regular.

\begin{theorem}\label{th:ut-dimVH}
For $\ula\in\Lambda_d(m,m,m)$ we have
$\dim (V_\ula)^{H_m} = \sum_{\a}
  \dim\big(V_{\la_1}^\a \ot V_{\la_2}^\a \ot V_{\la_3}^\a \big)^{\stab(\a)}$,
where the sum is over all
$\a\in\Par_m(d)$ such that $\a \preceq \la_1\meet\la_2\meet\la_3$.
\end{theorem}

The next result shows that the highest weights outside
the auxiliary semigroup of the unit tensors are very rare.

\begin{corollary}\label{cor:regmet}
(1) If there is a regular $\a \preceq \la_1\meet\la_2\meet\la_3$, then
$\ula\in S^o(\lan m \ran)$.

(2) If $\la_1,\la_2,\la_3$ are all regular, then $\ula\in S^o(\lan m \ran)$.
\end{corollary}

\begin{proof} (1) is an immediate consequence of Proposition~\ref{pro:pewe}
and Theorem~\ref{th:ut-dimVH}.
If $\la_1,\la_2,\la_3$ are all regular, then
$\bot_m(d) \preceq \la_i$ for $i=1,2,3$ by Lemma~\ref{le:unique-min}(1).
Now apply~(1).
\end{proof}

\begin{theorem}\label{th:no-lower-bound}
For any $\ula\in\Lambda^+_d(m,m,m)$ there exists $k\in\N$ such that
$\ula'\in S^o(\langle m+1 \rangle)$,
where
$\lambda_i = (\lambda_{i}^1,\ldots,\lambda_{i}^m)$ and
$\lambda'_i = (\lambda_{i}^1 +k,\ldots,\lambda_{i}^m +k,0)$.
\end{theorem}

\begin{proof}
Lemma~\ref{le:unique-min}(2) implies
$\bot_{m+1}(d +km) \preceq \la'_{1},\la'_{2},\la'_{3}$
for sufficiently large~$k$.
Now apply Corollary~\ref{cor:regmet}(1).
\end{proof}

Theorem~\ref{th:no-lower-bound} has severe consequences.
It tells us that for any tensor~$w$ of format $(m,m,m)$,
the trivial lower bound $\br(w)  > m$ is the best that can
be shown using $G_s$-obstructions!

\section{Matrix multiplication tensors}\label{se:mamu}

We fix complex vector spaces $U_i$ of dimension~$n_i$, 
put
$W_{12} := U_1^*\ot U_2,\
 W_{23} := U_2^*\ot U_3,\
 W_{31} := U_3^*\ot U_1$,
and consider the group
$$
 \Gc := \GL(U_1^*\ot U_2)\times \GL(U_2^*\ot U_3)\times\GL(U_3^*\ot U_1) .
$$
acting on $W :=W_{12}\ot W_{23}\ot W_{31}$.
We define the {\em matrix multiplication tensor} $\mamU\in W :=W_{12}\ot W_{23}\ot W_{31}$
as the tensor corresponding to the linear form
$$
 W^*\to\C,\  \ell_1\ot u_2\ot\ell_2\ot u_3\ot\ell_3\ot u_1 \mapsto \ell_1(u_1)\ell_2(u_2)\ell_3(u_3),
$$
obtained as the product of three contractions
($\ell_i\in U_i^*$ and $u_i\in U_i$).
To justify the naming we note that,
using the canonical isomorphisms
$\Hom(U_2,U_1)\simeq U_1\ot U_2^*$ and
$\Bil(U,V;W)\simeq U^*\ot V^* \ot W$,
one easily checks that $\mamU$ corresponds to the bilinear map
$$
\mamU\colon \Hom(U_2,U_1)\times\Hom(U_3,U_2)\to\Hom(U_3,U_1),\
 (\varphi,\psi) \mapsto \varphi\circ\psi
$$
describing the composition of linear maps
(note that we exchanged the order for the third factor:
$\Hom(U_3,U_1) \simeq U_3^*\ot U_1$).
If $U_i=\C^{n_i}$, then this bilinear map
corresponds to the multiplication of
$n_1\times n_2$ with $n_2\times n_3$ matrices.
In this case we shall write $\lan n_1,n_2,n_3 \ran =\mamU$.

\subsection{Stabilizer and stability}

We put
$\Sc := \GL(U_1)\times \GL(U_2) \times \GL(U_3)$
and consider the following morphism of groups
\begin{equation}\label{eq:mophi}
\Phi\colon \Sc \to \Gc,\
 (a_1,a_2,a_3)\mapsto
 ((a_1^*)^{-1} \ot a_2, (a_2^*)^{-1} \ot a_3, (a_3^*)^{-1} \ot a_1)
\end{equation}
with the kernel $\C^\ti\cdot\Id\simeq\C^\ti$.
Note that $\GL(U_i)$ acts on $U_i^*\ot U_i$ in the following way:
$$
 a_i\cdot (\ell_i \ot u_i) =
 \big((a_i^{-1})^* \ot a_i\big)(\ell_i \ot u_i) = (a_i^{-1})^*(\ell_i) \ot a_i(u)
 = (\ell_i\circ a_i^{-1}) \ot a_i(u_i) .
$$
Hence this action leaves the trace
$U_i^*\ot U_i\to\C, \ell_i\ot u_i \mapsto \ell_i(u_i)$ invariant.
This implies that the image of $\Phi$ is contained in the stabilizer $\Hc$ of $\mamU$.
In fact, equality holds.

\begin{proposition}\label{th:stab-mamu}
The stabilizer $\Hc$ of $\mamU$ equals the image of~$\Phi$.
In particular, $\Hc\simeq \Sc/\C^\ti$.
\end{proposition}

In the cubic case, this is a consequence of~\cite{deGroote:78}.
We remark that $\mamU$ is uniquely determined by its stabilizer up to a scalar.

\begin{proposition}\label{th:mamu-stable}
The matrix multiplication tensor $\mamU$ is stable.
\end{proposition}

This can be shown by Kempf's refinement~\cite{kempf:78}
of the Hilbert-Mumford criterion, cf.~\cite{meyer:06}.

\subsection{Representations}

Suppose that
$\la_{12}\in\Z^{n_1n_2}$ is a highest weight vector for $\GL(U_1^*\ot U_2)$ and
$\la_{23}\in\Z^{n_2n_3}$, $\la_{31}\in\Z^{n_3n_1}$
are highest weight vectors for $\GL(U_2^*\ot U_3)$ and $\GL(U_3^*\ot U_1)$, respectively.
Put $\ula=(\la_{12},\la_{23},\la_{31})$ and
consider the irreducible $\Gc$-module
$$
 V_{\ula} := V_{\la_{12}}(\GL(U_1^*\ot U_2))\ot V_{\la_{23}}(\GL(U_2^*\ot U_3))\ot
   V_{\la_{31}}(\GL(U_3^*\ot U_1)).
$$

\begin{theorem}\label{cor:H-inv-mamu}
Let $\la_{12},\la_{23},\la_{31}$ be partitions of~$d$
and $H$ be the stabilizer of~$\mamU$. Then
$$
 \dim (V_\ula)^{\Hc} \ =
  \sum_{\mu_1\vdash_{n_1}d,\mu_2\vdash_{n_2}d,\mu_3\vdash_{n_3}d}
  g(\la_{12},\mu_1,\mu_2)\cdot
  g(\la_{23},\mu_2,\mu_3)\cdot
  g(\la_{31},\mu_3,\mu_1) .
$$
\end{theorem}

\section{A few examples and computations}\label{se:comp}

\subsection{A family of $G$-obstructions}

We use the frequency notation 
$k_1^{e_1} k_2^{e_2}\cdots k_s^{e_s}$ 
to denote the partition of $\sum_i k_i e_i$ 
where $k_i$ occurs $e_i$ times. 


\begin{lemma}\label{le:G-obst}
We have 
$\ula_n := (2^{n^2} 0, 2^{n^2}0, (2n^2 -3)^1 1^3 0^{n^2-3}) 
 \in S(\langle n,n,n\rangle) \setminus S^o(\langle n^2 +1\rangle)$
for $n\ge 2$. This implies 
$\br(\lan n,n,n\ran) > n^2 +1$. 
\end{lemma}

\begin{proof} (Sketch)
1. We apply Theorem~\ref{th:ut-dimVH}. Put $N:=n^2$ and 
$\ula_n=(\la_1,\la_2,\la_3)$. 
The only partitions $\a\in\Par_{N+1}(2 N)$ smaller than $\la_1,\la_2,\la_3$ 
are $2^{N}0$ and $2^{N-1} 1^2$.  
A computation using the tableaux straightening algorithm from~\cite[p.110]{fult:97}
shows that
$(V_{\la_1}^\a \ot V_{\la_2}^\a \ot V_{\la_3}^\a \big)^{\stab(\a)}=0$
for both $\a$.
Proposition~\ref{pro:pewe} tells us that $\ula_n\not\in S^o(\langle n^2 +1\rangle)$. 

2. Using~\cite{rewe:94,rosa:01} one can show $g(\ula_n)=1$. Hence 
the highest weight vector $f\in\Oh(W)$ of weight~$\ula_n$
is uniquely determined up to a scalar. 
We explicitly constructed~$f$ and (guided by computer calculations) 
proved that $f(\langle n,n,n\rangle)\ne 0$.
Hence $\ula_n\in S(\langle n,n,n\rangle)$.
For the lower bound on $\br$ apply~\eqref{eq:strat}.
\end{proof}

\begin{remark}\label{re:mu}
1. Theorem~\ref{cor:H-inv-mamu} with $\mu_i=(2n)^n$, 
a positivity proof for the resulting Kronecker coefficients, 
and Proposition~\ref{pro:pewe} 
yield $\ula_n\in S^o(\langle n,n,n\rangle)$.
In order to guarantee $\ula_n\in S(\langle n,n,n\rangle)$
we currently know of no better way than 
to evaluate a highest weight vector at $\langle n,n,n\rangle$. 
In general, this becomes prohibitively costly for larger dimension formats.

2. Lemma~\ref{le:G-obst} yields $\br(\langle 2,2,2\rangle) > 5$.
It is known~\cite{land:06} that
$\br(\langle 2,2,2\rangle) = 7$.
So far we have been unable to reach the optimal lower bound 
by an obstruction.
\end{remark}

\subsection{Strassen's invariant}

Let $W=\C^m\ot\C^m\ot\C^3$, $m\ge 3$, and consider $\ula_m:=(3^m,3^m,m^3)$.
Strassen~\cite{stra:83} constructed an explicit invariant
$f_m\in\Oh(W)$ of highest weight $\ula_m^*$, that vanishes on all
tensors in $W$ with border rank at most $r=\lceil 3m/2\rceil -1$.
Hence $f_m(w)\ne 0$ implies $\br(w) > r$.

Let $\ula'_m\in\Lambda^+(r,r,r)$ be obtained from $\ula_m$ by appending
zeros.
It is tempting to conjecture that
$\ula'_m\not\in S(\langle r\rangle)$, because then Strassen's
implication would be a consequence of the existence of the obstruction
$\ula_m$.
Indeed, $f_m(w)\ne 0$ implies $\ula_m\in S(w)$ and, assuming the
conjecture,
$\ula_m\in S(w)\setminus S(\langle r\rangle)$ and thus $\br(w) >r$.
Unfortunately, the conjecture is already false for $m=4$!
An extensive computer calculation revealed the existence of 
$\tilde f_4\in \Oh(W)_{12}$ of highest weight $\ula_4^*$ and $g \in G$
such that $\tilde f_4 (g \langle 5\rangle) \neq 0$, which shows
$\ula'_4\in S(\langle 5\rangle)$.
Note $g(\ula_4)=2$. 

\section{Extension problem and nonnormality}\label{se:exten}

In order to advance, we need to study the difference between $S(w)$ and $S^o(w)$.
Let $W$ be of format $\um$ and $w\in W$ be stable.
If $\ula\in S^o(w)$, then Proposition~\ref{th:stabSs} implies that there
exists $k\in\Z$ such that $\ula + k\e_\um\in S(w)$,
where $\e_\um  =(\e_{m_1},\e_{m_2},\e_{m_3})$.
It is of interest to know the smallest such~$k$.
Below we will see that~$k$ can be given a geometric interpretation in terms of
the problem of extending regular functions from $Gw$ to $\ol{Gw}$.

We call the group morphism
$\det\colon G\to\C^\times, (g_1,g_2,g_3)\mapsto \det g_1\,\det g_2\,\det g_3$
the determinant on~$G$.
In the following we will assume that $\e_\um\in S^o(w)$.
By Proposition~\ref{pro:pewe} this is equivalent to
$\det g =1$ for all $g\in\stab(w)$.
We note that this condition is satisfied for $w=\langle n,n,n\rangle$
due to Proposition~\ref{th:stab-mamu}.

If $\e_\um\in S^o(w)$, then $\det$ induces the well defined regular function
$\det_w\colon Gw \to \C, gw \mapsto \det g$.

\begin{theorem}\label{th:extend}
Suppose that $w\in W$ is a stable tensor and $\e_\um\in S^o(w)$.
\begin{enumerate}
\item Then $w$ has the cubic format $(m,m,m)$.
\item The extension of $\det_w$ to $\ol{Gw}$
with value~$0$ on the boundary $\ol{Gw}\setminus Gw$
is continuous in the $\C$-topology.
\item $\det_w$ is not a regular function
on $\ol{Gw}$ if $m>1$.
\item $\ol{Gw}$ is not a normal variety if $m>1$.
\item For all highest weight vectors~$f\in\Oh(Gw)$
we have $(\det_w)^k f \in \Oh(\ol{Gw})$ for some $k\in\Z$.
\end{enumerate}
\end{theorem}

We can also show a variant of this result with $\det$ replaced by $\det^2$.
This is of interest since
$(\det g)^2 = 1$ for all $g\in\stab(\langle m\rangle)$,
cf.~Proposition~\ref{pro:stabut}.

\begin{corollary}\label{cor:nn-mamu}
(1) The orbit closure of the matrix multiplication tensor $\langle n,n,n\rangle$
is not normal if $n>1$.
(2) The orbit closure of the unit tensor $\langle m\rangle$ is not normal if $m \ge 5$.
\end{corollary}

The nonnormality of these orbit closures indicates that the extension problem is delicate.
Kumar~\cite{kuma:10} recently obtained similar conclusions for the orbit closures
of the determinant and permanent by different methods.

We also make the following general observation.

\begin{proposition}\label{pro:codim1}
Suppose that $w\in W$ is stable. Then $\stab(w)$ is reductive,
$Gw$ is affine. Further, $\overline{Gw}\setminus Gw$ is either empty or of
pure codimension one in $\overline{Gw}$.
\end{proposition}

\section{Moment polytopes}\label{se:mp}

Since the semigroups $S(w)$ seem hard to determine, one may take
a coarser viewpoint, as already suggested by Strassen~\cite[Eq.~(57)]{stra:05}.
We set $\Delta_{\um} := \Delta_{m_1}\times \Delta_{m_2}\times \Delta_{m_3}$,
where
$\Delta_m:=\{ x\in \R^m \mid x_1\ge \ldots \ge x_m\ge 0,\, \sum_i x_i=1\}$.

\begin{definition}\label{def:mp}
The {\em moment polytope} $P(w)$ of a tensor $w\in W$ is defined
as the closure of the set
$\big\{ \frac1{d}\,\ula \mid d>0,\, \ula \in S(w) \cap \Lambda^+_d(\um)\big\}$.
\end{definition}

Note that $P(w)\subseteq \Delta_\um$ is a polytope since $S(w)$
is a finitely generated semigroup. We have
$$
 \ol{Gw}\subseteq\ol{Gv} \Longrightarrow S(w)\subseteq S(v)
 \Longrightarrow P(w) \subseteq P(v)
$$
Hence exhibiting some point in $P(w)\setminus P(\langle m \rangle)$
would establish the lower bound $\br (w) >m$.

The moment polytope of a generic tensor $w$ of format~$\um$
equals the {\em Kronecker polytope}~$P(\um)$, which is defined as the closure of
$\{ \frac1{d}\ula \mid \ula \in K(\um), \ula \in \Lambda^+_d(\um)\}$,
compare Lemma~\ref{pro:Sgen}.
This complicated polytope has been the object of several recent investigations
\cite{besj:00,klya:04,ress:09,buci:09} and $P(\um)$ is by now understood to a certain extent.
We remark that the Kronecker polytope $P(\um)$ is closely related
to the quantum marginal problem of quantum information theory,
cf.~\cite{ChrMit06,klya:04}.

Let $u_m:=(1/m,\ldots,1/m)\in\Delta_m$ denote the uniform distribution
and put $u_\um:=(u_m,u_m,u_m)$.
The following follows, e.g., from~\cite[Satz~11]{stra:05}.

\begin{lemma}\label{le:cp-mp}
We have $u_\um \in P(w)$ both for
$w=\langle m \rangle$ and $w=\langle n,n,n \rangle$, $m=n^2$. 
\end{lemma}

Resolving the following question seems of great relevance.

\begin{problem}
Determine the moment polytopes of unit tensors and matrix multiplication tensors.
\end{problem}

Replacing $S(w)$ by $S^o(w)$ in the definition of $P(w)$ we obtain the larger
polytope $P^o(w)$.

\begin{theorem}\label{th:mpzero}
We have
$P^o(\langle m \rangle)) = \Delta(m,m,m)$ and
$P^o(\langle n,n,n \rangle)) = \Delta(n^2,n^2,n^2)$.
\end{theorem}

The statement for the unit tensors is an easy consequence of
Corollary~\ref{cor:regmet}(2). The second statement relies on
Theorem~\ref{cor:H-inv-mamu} and \cite{buci:09}.

\begin{lemma}\label{le:umint}
Let $w$ be stable and suppose that $u_\um\in P(w)$.
Then there exists $\d>0$ such that for all $\ux\in P^o(w)$
and all $0\le t\ \le \d$ we have $t\ux + (1-t)u_\um \in P(w)$.
\end{lemma}

By combining Theorem~\ref{th:mpzero} with Lemma~\ref{le:cp-mp},
Lemma~\ref{le:umint}, and the stability of the unit
and matrix multiplication tensors, we obtain the following result.

\begin{corollary}\label{cor:mp-int}
For both $w=\langle m \rangle$ and $w=\langle n,n,n \rangle$, $m=n^2$, 
$u_{\um}$ is an interior point of $P(w)$
relative to the affine hull of $\Delta_{\um}$.
In particular, $\dim P(w)= \dim \Delta_{\um}$.
\end{corollary}

\section{Outline of some proofs}

\subsection{Proof of Theorem~\ref{th:ut-dimVH}}

Let $\ula\in\Lambda_d(m,m,m)$. The weight decomposition
$$
 V_\ula = V_{\la_1} \ot V_{\la_2}\ot V_{\la_3} =
 \bigoplus_{\a,\b,\g}
 V_{\la_1}^{\a} \ot V_{\la_2}^{\b} \ot V_{\la_3}^{\g}
$$
yields
$(V_\ula)^{\Tc_m} = \bigoplus_{\a,\b,\g}
 \big(V_{\la_1}^{\a} \ot V_{\la_2}^{\b} \ot V_{\la_3}^{\g}\big)^{\Tc_m}$, 
cf. Proposition~\ref{pro:stabut}.
We claim that
$$
 \big(V_{\la_1}^{\a} \ot V_{\la_2}^{\b} \ot V_{\la_3}^{\g}\big)^{\Tc_m}
 = \left\{\begin{array}{ll}
   V_{\la_1}^{\a} \ot V_{\la_2}^{\a} \ot V_{\la_3}^{\a} & \mbox{ if $\a=\b=\g$,} \\
   0 & \mbox{ otherwise.}
   \end{array}\right.
$$
Indeed, let $v\in (V_{\la_1}^{\a}\otimes V_{\la_2}^{\b} \otimes V_{\la_3}^{\g})^{\Tc_m}$
be nonzero.
For $t=(\diag(a),\diag(b),\diag(c)\in \Tc_m$
we obtain
$v = tv = a^{\a} b^{\b} c^{\g} v = a^{\a-\g} b^{\b -\g} v$,
using $a_ib_ic_i=1$.
Since the $a_i,b_i\in\C^\times$
are arbitrary, we infer $\a=\b=\g$.
The argument can be reversed.

We put now
$$
 A :=\big\{\a\in\Z^m \ \big|\ |\a| =d,\, \a \preceq \la_1\meet\la_2\meet\la_3 \big\},\quad
 M^\a := V_{\la_1}^{\a} \ot V_{\la_2}^{\a} \ot V_{\la_3}^{\a}
$$
and note that $M^\a\ne 0$ for all $\a\in A$.
We have just seen that
$(V_\ula)^{\Tc_m} = \oplus_{\a\in A} M^\a$.

The set $A$ is invariant under the $S_m$-action and
its orbits intersect $\Par_m(d)$ in exactly one partition.
We note that $\pi M^\a = M^{\pi\a}$ for $\pi\in S_m$.
Let $\mathcal{B}$ denote the set of orbits and put
$M_B:= \oplus_{\a\in B} M^\a$ for $B\in\mathcal{B}$.
Then
$(V_\ula)^{\Tc_m} =\oplus_{B\in\mathcal{B}} M_B$.
Proposition~\ref{pro:stabut} tells us
$H_m=\Tc_m S_m$ and hence
$$
(V_\ula)^{H_m} = \big((V_\ula)^{\Tc_m}\big)^{S_m}
= \bigoplus_{B\in\mathcal{B}} (M_B)^{S_m}
$$
using that the $M_B$ are $S_m$-invariant.
In order to complete the proof it suffices to show that
\begin{equation*}\label{eq:dimMS}
 \dim (M_B)^{S_m} = \dim (M^\a)^{\stab(\a)}\quad
 \mbox{for $B=S_m\a$, $\a\in A\cap \Par_m(d)$.}
\end{equation*}

For proving this, we fix $\a\in A\cap \Par_m(d)$ and write $H:= \stab(\a)$.
Let $\pi_1,\ldots,\pi_t$ be a system of representatives for the left cosets of~$H$ in $S_m$
with $\pi_1=\Id$. So
$S_m = \pi_1 H \cup \cdots \cup \pi_t H$.
Then the $S_m$-orbit of~$\a$ equals
$S_m\a=\{\pi_1\a,\ldots,\pi_t\a\}$.
Consider
$$
 M_B = \bigoplus_{j=1}^t \pi_j M^\a
$$
and the corresponding projection $p\colon M_B\to M^\a$.
Suppose that $v=\sum_j v_j\in (M_B)^{S_m}$ with $v_j\in \pi_j M^\a$.
Since the spaces
$\pi_1 M^\a,\ldots,\pi_t M^\a$
are permuted by the action of $S_m$,
we derive from
$v = \pi_k v=\sum_j \pi_k v_j$ that
$v_j=\pi_j v_1$. Moreover, since $\s\in H$ fixes $M^\a$ and permutes
the spaces $\pi_2 M^\a,\ldots,\pi_t M^\a$, we obtain
$\s v_1= v_1$.
Therefore, $(M_B)^{S_m}\to (M^\a)^H, v \mapsto p(v) = v_1$
is well defined and injective.
We claim that this map is also surjective.

For showing this, let $v_1 \in (M^\a)^H$, set $v_j := \pi_j v_1$, and put
$v:=\sum_j v_j$. Clearly, $p(v)=v_1$.
Fix $\s\in H$ and $i$.
For any $j$ there is a unique $k=k(j)$ such that
$\s\pi_i\pi_j H = \pi_k H$. Moreover,
$j\mapsto k(j)$ is a permutation of $\{1,\ldots,t\}$.
Using the $H$-invariance of~$v_1$ we obtain that
$\s\pi_i v_j = \s\pi_i \pi_j v_1 = \pi_{k} v_1 =v_k$.
Therefore $\s\pi_i v = \sum_k v_k = v$.
Thus $v\in (M_B)^{S_m}$.
\qed

\subsection{Proof of Theorem~\ref{cor:H-inv-mamu}}

The group morphisms
\begin{eqnarray*}
\Gamma_{12}\colon\GL(U_1^*)\ti\GL(U_2)\to \GL(U_1^*\ot U_2),& (a^*,b)\mapsto a^*\ot b \\
\Gamma_{23}\colon\GL(U_2^*)\ti\GL(U_3)\to \GL(U_2^*\ot U_3),& (b^*,c)\mapsto b^*\ot c \\
\Gamma_{31}\colon\GL(U_3^*)\ti\GL(U_1)\to \GL(U_3^*\ot U_1),& (c^*,a)\mapsto c^*\ot a
\end{eqnarray*}
combine to a morphism $\Gamma\colon\Pc\to\Gc$, where $\Pc$ denotes the group
$$
 \Pc := \GL(U_1^*)\ti\GL(U_2)\ti\GL(U_2^*)\ti\GL(U_3)\ti\GL(U_3^*)\ti\GL(U_1).
$$
Moreover, we have the group morphisms
$$
\Lambda_i\colon \GL(U_i)\to\GL(U_i^*)\ti\GL(U_i),
 a_i\mapsto ((a_i^*)^{-1}, a_i)
$$
combining to a morphism (note the permutation)
$$
 \Lambda\colon \Sc \to \Pc,\,(a_1,a_2,a_3)\mapsto
 \big((a_1^*)^{-1}),a_2,(a_2^*)^{-1},a_3,(a_3^*)^{-1},a_1\big).
$$
We have thus factored the morphism $\Phi\colon\Sc\to\Gc$ as
$\Phi=\Gamma\circ\Lambda$, cf.~\eqref{eq:mophi}.
Proposition~\ref{th:stab-mamu} states that $\Hc=\im\Phi$.
In order to determine $\dim (V_\ula)^{\Hc}$,
we first describe the splitting of $V_\ula$ into irreducible
$\Pc$-modules with respect to $\Gamma$ and then,
in a second step, extract their $\Sc$-invariants.

For the first step, note that,
upon restriction with respect to $\Gamma_{12}$,
we have the decomposition
$$
 V_{\la_{12}}(\GL(U_1^*\ot U_2)) = \bigoplus_{\mu_1,\tilde{\mu}_2}
  g(\la_{12},\mu_1,\tilde{\mu}_2) V_{\mu_1}(\GL(U_1^*))\ot V_{\tilde{\mu}_2}(\GL(U_2)),
$$
where the sum is over all partitions $\mu_1\vdash_{n_1} d$, $\tilde{\mu}_2\vdash_{n_2} d$.
For this characterization of the Kronecker coefficients $g$
see~\cite[(7.221), p. 537]{StanleyVol2}.
Similarly,
$$
 V_{\la_{23}}(\GL(U_2^* \ot U_3)) = \bigoplus_{\mu_2,\tilde{\mu}_3}
  g(\la_{23},\mu_2,\tilde{\mu}_3) V_{\mu_2}(\GL(U_2^*))\ot V_{\tilde{\mu}_3}(\GL(U_3)),
$$
$$
 V_{\la_{31}}(\GL(U_3^*\ot U_1)) = \bigoplus_{\mu_3,\tilde{\mu}_1}
  g(\la_{31},\mu_3,\tilde{\mu}_1) V_{\mu_3}(\GL(U_3^*))\ot V_{\tilde{\mu}_1}(\GL(U_1)),
$$
where the sums are over all $\mu_2\vdash_{n_2} d$,
$\tilde{\mu}_3\vdash_{n_3} d$ and $\mu_3\vdash_{n_3} d$, $\tilde{\mu}_1\vdash_{n_1} d$,
respectively.
This describes the splitting of $V_\ula$ into irreducible
$\Pc$-modules with respect to $\Gamma$.

For the second step we note that
$V_{\mu_i}(\GL(U_i^*)) \simeq V_{\mu^*_i}(\GL(U_i))$,
when we view the left hand side as a $\GL(U_i)$-module
via the isomorphism
$\GL(U_i)\to\GL(U_i^*),a_i\mapsto (a_i^*)^{-1}$.

As a consequence of the Littlewood-Richardson rule~\cite{StanleyVol2,fult:97}
we obtain (compare~\cite[Eq.~(11), p.149]{fult:97})
\begin{equation}\label{eq:LWRdual}
\dim \big(V_{\mu_i^*}(\GL(U_i)) \ot V_{\tilde{\mu}_i}(\GL(U_i)) \big)^{\GL(U_i)}
= \left\{ \begin{array}{ll} 1 & \mbox{ if $\mu_i = \tilde{\mu}_i$},\\
                                           0  & \mbox{ otherwise.}
   \end{array}\right.
\end{equation}
We conclude that
$$
 \dim (V_\ula)^{\Hc} =  \bigoplus_{\mu_1,\mu_2,\mu_3}
  g(\la_{12},\mu_1,\mu_2)\, g(\la_{23},\mu_2,\mu_3)\, g(\la_{31},\mu_3,\mu_1)
$$
as claimed.
\qed

\bibliographystyle{amsplain}
\bibliography{gct-refs}

\newpage

\section*{Appendix: Section~\ref{se:GCTprog}}

\addtocounter{section}{1}
\setcounter{subsection}{0}
\setcounter{equation}{0}

\subsection{Highest weight vectors}\label{se:HWV}

For the following general facts see~\cite{good-wall:09,hump:75,kraf:84}. 
In general, let $G$ be a reductive group and fix a Borel subgroup with corresponding 
maximal unipotent subgroup~$U$ and maximal torus~$T$. 
Let $M$ be a rational $G$-module.  
By a {\em highest weight vector} in $M$ of weight~$\la\in\Lambda^+_G$ 
we understand a $U$-invariant weight vector in $M$ of weight~$\la$. 
These vectors (including the zero vector) form a linear subspace  of $M$
that we shall denote by $\HWV_\la (M)$. 
It is known that $M$ is irreducible iff $\HWV_\la (M)$ is one-dimensional. 
If $\varphi\colon M\to N$ is a surjective $G$-module morphism, 
then $\varphi(\HWV_\la(M)) = \HWV_\la(N)$. 

\subsection{Schur-Weyl duality}

For the following known facts see~\cite{fuha:91}. 
Recall that $[\la]$ denotes the irreducible module of the symmetric group~$S_d$ 
associated with a partition $\la\vdash d$. 
Let $V$ denote a vector space of dimension~$m$. 
The group $S_d$ acts on $V^{\ot d}$ by permutation.
We define the $\la$th {\em Schur-Weyl module} by 
$$
 S_\la (V) := \Hom_{S_d} ([\la], V^{\ot d}) .
$$
Note that $S_\la (V)$ becomes a $\GL(V)$-module in a natural way. 
It is well known that $S_\la (V)=0$ if $\la$ has more than $m$ parts.
Otherwise, $S_\la (V)$ is an irreducible $\GL(V)$-module,  
and all the irreducible $\GL(V)$-modules of degree~$d$ 
are isomorphic to $S_\la(V)$ for some $\la\vdash_m d$. 

A linear map $\varphi\colon V\to W$ induces the linear map 
$S_\la(\varphi)\colon S_\la V\to S_\la W, \a\mapsto \varphi^{\ot d} \a$
and the functorial property 
$S_\la(\psi\phi) = S_\la(\psi) S_\la (\phi)$ clearly holds. 

We have natural injective maps 
$[\la]\ot S_\la(V) \to V^{\ot d}, x\ot \a \mapsto \a(x)$, 
which yield the following canonical injective morphism of $S_d\ti \GL(V)$-modules: 
\begin{equation}\label{eq:SW-duality}
\bigoplus_{\la\vdash_m d} [\la] \ot S_\la(V) \to V^{\ot d} .
\end{equation}
Schur-Weyl duality states that this map is surjective. 
In particular, ~\eqref{eq:SW-duality} yields the isotypic decomposition of 
$V^{\ot d}$ with respect to the action of $S_d\ti \GL(V)$.

Depending on the context, we also use the notation 
$V_\la(\GL(V)) := S_\la(V)$ for emphasizing the dependence on the group.

\subsection{Decomposition of $\Oh(W)$}\label{se:decompOW}

Let $W_i$ be a vector space of dimension~$m_i$ for $i=1,2,3$, 
put $W:=W_1\ot W_2\ot W_3$, and $G = \GL(W_1)\times \GL(W_2)\times \GL(W_3)$.
Applying Schur-Weyl duality~\eqref{eq:SW-duality} to $W_i^*$ and 
taking the tensor product yields the isotypic decomposition with respect the action of
of the group $S_d\times S_d \times S_d\times G$: 
\begin{equation*}\begin{split}
 (W^*)^{\otimes d} &\simeq (W_1^*\otimes W_2^*\otimes W_3^*)^{\otimes d} \simeq
 (W_1^*)^{\otimes d} \otimes (W_2^*)^{\otimes d} \otimes (W_3^*)^{\otimes d} \\
 & \simeq \bigoplus_{\ula\in\Lambda^+_d(\um)} [\la_1] \otimes [\la_2]\otimes [\la_3]
  \otimes S_{\ula}(W^*) ,
\end{split}
\end{equation*}
where we have put 
$S_\ula (W^*) := S_{\la_1}(W_1^*) \ot S_{\la_2}(W_2^*) \ot S_{\la_3}(W_3^*)$,
which is the same as  $V_\ula(G)^*$ 
since $S_{\la_i}(W_i^*) \simeq (S_{\la_i}(W_i))^*$. 
We interpret $S_d$ as a subgroup of $S_d\ti S_d\ti S_d$ 
with respect to the diagonal embedding $\pi \mapsto (\pi,\pi,\pi)$ and 
note that 
$$
 \Oh(W)_d = \Sym^d W^* = \big((W^*)^{\otimes d}\big)^{S_d} .
$$
We therefore arrive at the following canonical isomorphism of $G$-modules: 
\begin{equation}\label{eq:decompOW}
  \Oh(W)_d = \Sym^d W^*
  \simeq
  \bigoplus_{\ula\in\Lambda^+_d(\um)}
  \big([\la_1] \otimes [\la_2] \otimes [\la_3]\big)^{S_d} \otimes S_{\ula}(W ^*)  .
\end{equation}
In particular, we obtain 
\begin{equation}\label{eq:multOW}
 g(\ula) = \dim    \big([\la_1] \otimes [\la_2] \otimes [\la_3]\big)^{S_d} 
             = \mult(V_\ula(G)^*,\Oh(W)_d) 
\end{equation}
by the definition of Kronecker coefficients. 


\subsection{Proof of Lemma~\ref{pro:Sgen}}

The restriction
$\res_{d,w}\colon \Oh(W)_d \to \Oh(\ol{Gw})_d$
is a surjective morphism of $G$-modules,
for each degree $d\in\N$.
Hence $S(w)\subseteq K(\um)$ follows from~\eqref{eq:multOW}.

Let $\ula_1,\ldots,\ula_s$ be a list of generators of the semigroup $K(\um)$.
By~\eqref{eq:multOW},
there exists a highest weight vector $F_i\in\Oh(W)_{d_i}$ of weight~$\ula_i^*$,
for each~$i$.
Assume $F_1(w)\cdots F_s(w)\ne 0$.
Then $\res_{d_i,w}(F_i)$ is a highest weight vector of weight~$\ula_i^*$
in $\Oh(\ol{Gw})_{d_i}$, hence $\ula_i \in S(w)$.
Since the generators $\ula_1,\ldots,\ula_s$ of $K(\um)$
are all contained in $S(w)$, we get $K(\um) \subseteq S(w)$.
So we have shown that the nonempty open set
$\{ w\in W \mid F_1(w)\cdots F_s(w)\ne 0 \}$
is contained in
$\{w\in W \mid K(\um) \subseteq S(w)\}$.

\subsection{Inheritance}

For this section compare~\cite{lama:04,blmw:09}. 
We first recall our standard notations: 
Let $W := W_1\ot W_2 \ot W_3$ and $W':=W'_1\ot W'_2 \ot W'_3$ be of 
format $\um$ and $\um'$, respectively. 
Let $\iota_i\colon W_i\hookrightarrow W'_i$ be inclusions of vector spaces and 
choose linear projections $p_i\colon W'_i\to W_i$ such that $p_i\, \iota_i =\Id_{W_i}$. 
For $\ula\in\Lambda^+_d(\um')$ the dual maps 
$\iota_i^*\colon W_i'^* \to W_i^*$ induce maps
$S_{\la_i}(\iota_i^*)\colon S_{\la_i}(W_i'^*) \to S_{\la_i}(W_i^*)$.
These maps are surjective since by the functoriality we have 
$S_{\la_i}(\iota_i^*)\, S_{\la_i}(p_i^*) = \Id$.
We further set 
$\iota :=\iota_1\ot \iota_2\ot \iota_3$, 
$p :=p_1\ot p_2\ot p_3$, and 
$S_{\ula}(\iota^*) := S_{\la_1}(\iota_1^*) \ot S_{\la_2}(\iota_2^*) \ot S_{\la_3}(\iota_3^*)$.
Note that 
$S_{\ula}(\iota^*)\colon S_{\ula}(W'^*) \to S_{\ula}(W*)$
is surjective.

Let $\res_d\colon\Oh(W')_d \to \Oh(W)_d$ denote the restriction of regular functions.
Equation~\eqref{eq:decompOW} yields the following commutative diagram 
(writing equalities for canonical isomorphisms) 
\begin{equation}\label{eq:comdia}
\begin{array}{cclc}
\Oh(W')_d & = & \bigoplus_{\ula\in\Lambda^+_d(\um')} & \Sigma_\ula \otimes S_{\ula}(W'^*) \\
{\scriptstyle \res_d} \downarrow &  & & {\scriptstyle \Id\ot S_\ula(\iota^*)}\downarrow \\[1ex]
\Oh(W)_d & = & \bigoplus_{\ula\in\Lambda^+_d(\um)} & \Sigma_\ula \otimes S_{\ula}(W^*) ,
\end{array}
\end{equation}
with the vector spaces $\Sigma_\ula:= \big([\la_1] \otimes [\la_2] \otimes [\la_3]\big)^{S_d}$.
Note that the vertical arrows are surjective. 

Recall that 
$G := \GL(W_1)\ti\GL(W_2)\ti\GL(W_3)$ and 
$G' := \GL(W'_1)\ti\GL(W'_2)\ti\GL(W'_3)$ 
and let $I(Z)$ denote the vanishing ideal of an affine variety $Z$.
We notationally identify $w\in W$ with it image $w'$ under the 
inclusion $\iota\colon W\to W'$. 

\begin{lemma}\label{le:oc}
1. For $w\in W$ we have 
$p(\overline{G'w}) = \overline{Gw}$. 

2. We have $\res(I(\overline{G'w})) = I(\overline{Gw})$ 
for the restriction $\res\colon\Oh(W')\to\Oh(W)$. 
\end{lemma}

\begin{proof}
1. Put $F_i := \ker p_i$ so that  $W_i'=W_i\oplus F_i$ and $p_i$ is the projection 
onto $W_i$ along~$F_i$. 
Let $g'=(g'_1,g'_2,g'_3) \in G'$ and consider
$g_i:=p_i\, g'_i \iota_i$. Note that $g_i\colon W_i\to W_i$ 
is a linear map that may be noninvertible. 
It is easy to check that 
$p(g'w) = (g_1\ot g_2 \ot g_3) w$. 
This implies the first assertion $p(\overline{G'w}) = \overline{Gw}$. 

2. For $f\in I(\overline{Gw})$ define $F := f\, p\in \Oh(W')$. Then 
$\res(F) = f$ and  part~(1) implies that $F$ vanishes on $\overline{G'w}$. 
This shows $\res(I(\overline{G'w})) \supseteq I(\overline{Gw})$.
The other inclusion is obvious.
\end{proof}

We analyze now the $G'$-submodule $I(\overline{G'w})_d \subseteq \Oh(W')_d$. 
The isotypic decomposition~\eqref{eq:decompOW} implies that for 
each $\ula\in\Lambda^+_d(\um')$ there exists a linear subspace 
$J_\ula(w) \subseteq  \Sigma_\ula$ such that
$$
  I(\overline{G'w})_d \ = \bigoplus_{\ula\in\Lambda^+_d(\um')} J_\ula(w) \otimes S_{\ula}(W'^*) . 
$$
By Lemma~\ref{le:oc}(2) and the commutative diagram~\eqref{eq:comdia}
we have 
$$
 I(\overline{Gw})_d  = \res_d(I(\overline{Gw'})_d) \ =  
 \bigoplus_{\ula\in\Lambda^+_d(\um)} J_\ula(w) \otimes S_{\ula}(W^*) . 
$$
Since  
$\Oh(\overline{Gw})_d = \Oh(W)_d/I(\overline{Gw})_d$, 
this implies, for $\ula\in\Lambda^+_d(\um)$, 
\begin{equation}\label{eq:samemult}
\mult(S_{\ula}(W^*),\Oh(\overline{Gw})) = \dim \Sigma_\ula - \dim J_\ula(w)
 = \mult(S_{\ula}(W'^*),\Oh(\overline{G'w})) .
\end{equation}

The case $\ula\in\Lambda^+_d(\um')\setminus \Lambda^+_d(\um)$ 
is covered by the following lemma.

\begin{lemma}\label{le:Jall}
We have $J_\ula(w) = \Sigma_\ula$  and hence 
$\mult(S_{\ula}(W'^*),\Oh(\overline{G'w})) =0$ 
if $\ula\in\Lambda^+_d(\um')\setminus \Lambda^+_d(\um)$. 
\end{lemma}

\begin{proof}
Let $M\subseteq \Oh(W')_d$ be a submodule such that $M\simeq S_{\ula}(W'^*)$. 
We need to show that 
$M\subseteq I(\overline{G'w})$ if $\ula\in\Lambda^+_d(\um')\setminus \Lambda^+_d(\um)$.  

It is convenient to prove the contraposition. So suppose that $M$ is not contained in $I(\overline{G'w})$.
Then there exists $F\in M$ such that $F(w)\ne 0$. 
(Indeed, by assumption there exist $F\in M$ and $g\in G'$ such that 
$0\ne F(g^{-1}w) =(gF)(w)$. Just replace $F$ by $gF\in M$.) 

We may assume that $F$ is a weight vector and 
let $-\ua = -(\a_1,\a_2,\a_3)$ denote its weight.
Write $\a_i=(\a_{i1},\ldots,\a_{im'_i})$.  
Let $t=(t_1,t_2,t_3)\in G'$ with $t_i= \diag(t_{i1},\ldots,t_{i m_i'})$
such that $t_{ij}=1$ for $1\le j \le m_i$.
Then $t w = w$ and we have 
$$
 \mbox{$F(w) = F(tw) =  (t^{-1} F)(w) = \prod_{i=1}^3\prod_{j =m_i +1}^{m_i'} t_{ij}^{\a_{ij}} F(w)$.}
$$
Since $t_{ij}\in\C^\times$ is arbitrary for $j>m_i$,
it follows that $\a_{ij}=0$ for $j>m_i$.

Since $\ua$ is a weight of $M^*\simeq S_{\ula}(W')$ and 
$\ula$ is the highest weight of~$M$, we have $\ua\preceq \ula$.
From $\a_{ij}=0$ for $j>m_i$ it follows that 
$\la_i$ has at most $m_i$ parts. 
This means $\ula\in\Lambda^+_d(\um)$.
\end{proof}

Equation~\eqref{eq:samemult} and Lemma~\ref{le:Jall}
imply Proposition~\ref{th:inheritance}.

\subsection{Decomposition of ring of regular functions on orbits}\label{se:peter-weyl}

Consider the action of $G\times G$ on $G$ defined by
$(g_1,g_2)\cdot g := g_1\, g\, g_2^{-1}$.
This induces the following action of $G\times G$ on~$\Oh(G)$:
$$
 ((g_1,g_2)\cdot f)(g) := f(g^{-1}\,g\, g_2) \quad\mbox{where $g_1,g_2,g\in G, f\in \Oh(G)$}.
$$
The usual ``left action'' of $G$ on $\Oh(G)$ is obtained
by the embedding
$G\hookrightarrow G\times G,\,g_1\mapsto (1,\Id)$.
But note that we also have a ``right action'' of $G$
given by
$G\hookrightarrow G\times G,\,g_2\mapsto (1,g_2)$.

We now state the fundamental algebraic Peter Weyl Theorem
for the group $G$, cf. \cite[p. 93]{kraf:84} or \cite{good-wall:09}.
(This result actually holds for any reductive group $G$.)

\begin{theorem}\label{th:peweorig}
The isotypic decomposition of $\Oh(G)$ as a
$G\times G$-module is given as
$$
 \Oh(G) = \bigoplus_{\ula\in\Lambda^+_G} V_\ula(G)^*\otimes V_\ula(G) .
$$
Here the $G\times G$-module structure on $V_\ula(G)^*\otimes V_\ula(G)$
is to be interpreted as
$(g_1,g_2)\,(\ell\otimes v) := g_1\ell \otimes g_2 v$
for $\ell\in V_\ula(G)^*$ and $v\in V_\ula(G)$.
\end{theorem}

\subsection{Proof of Proposition~\ref{pro:pewe}}

The stabilizer $H$ of $w$ is a closed subgroup of $G$.
It is possible to give $G/H=\{gH \mid g\in G\}$
the structure of an algebraic variety such that
$G\to G/H$ is a morphism of varieties satisfying the universal property of
quotients (cf.~\cite[Chap.~12]{hump:75}.
This implies that $G/H$ is isomorphic to $Gw$ as a variety
and this morphism is $G$-equivariant.

The morphism of coordinate rings $\Oh(G/H) \to \Oh(G)$
induced by $G\to G/H$ is injective, and maps to the subring
$$
 \Oh(G)^H = \{f\in\Oh(G) \mid \forall g\in G, h\in H\ f(gh) = f(g)\}
$$
of $H$-invariant functions with respect to the right action of $H$ on $\Oh(G)$.
We note that $\Oh(G)^H$ is a $G$-submodule of $\Oh(G)$ with respect
to the left action of $G$. Moreover, $\Oh(G/H) \to \Oh(G)^H$ is
$G$-equivariant.
The universal property of quotients implies the surjectivity of
$\Oh(G/H) \to \Oh(G)^H$.
So we have shown that $\Oh(Gw)$ is isomorphic to the $G$-module $\Oh(G)^H$.

Theorem~\ref{th:peweorig} implies that
$$
 \Oh(G)^H = \bigoplus_{\ula\in\Lambda^+_G} V_\ula(G)^*\otimes V_\ula(G)^H .
$$
Hence
$\mult(V_\ula(G)^*, \Oh(G)^H)  = \dim V_\ula(G)^H$, which
completes the proof.
\qed


\section*{Appendix: Section~\ref{se:unit}}

\subsection{Stabilizers of associative algebras}\label{se:stab-ass}

Let $\Bil(U,V;W)$ denote the space of bilinear maps $U\ti V\to W$,
where $U,V,W$ are finite dimensional vector spaces.
The group $G=\GL(U)\ti \GL(V)\ti \GL(W)$ acts on
$\Bil(U,V;W)$ via
$(\a,\b,\g) \cdot \varphi := \g\,\varphi(\a^{-1}\ti\b^{-1})$.
By definition, $\GL(U)$ acts on the dual module $U^*$ via
$\a \cdot \ell := (\a^{-1})^* (\ell)$ for $\a\in\GL(U)$, $\ell \in U^*$.
It is straightforward to check that the canonical isomorphism
$U^*\ot V^*\ot W\to \Bil(U,V;W)$ is $G$-equivariant.
Hence we obtain the following result.

\begin{lemma}\label{le:bilstab}
Let $\varphi\in\Bil(U,V;W)$ and $w\in U^*\ot V^*\ot W$ be the
corresponding tensor. Then
$$
 \stab(w)  =\big\{ ((\a^{-1})^*, (\b^{-1})^*, \g) \mid \forall u,v \
    \varphi(\a(u),\b(v)) =\g(\varphi(u,v)) \big\} .
$$
\end{lemma}

Now let $A$ be a finite dimensional associative $\C$-algebra with~$1$.
Its multiplication map $A\ti A\to A$ corresponds to a tensor
$w_A\in A^*\ot A^* \ot A$.
We denote by $A^\ti$ the unit group of $A$ and by $\Aut A$ its group
of algebra automorphisms. For $a\in A$ we denote by
$L_a\colon A\to A, x\mapsto ax$ the left multiplication with~$a$.
Similarly, $R_a$ denotes the right multiplication with~$a$.

The following observation goes back to~\cite{deGroote:78}.

\begin{lemma}\label{le:deGroote}
We have
$$
 \stab(w_A) = \Big\{
  \big( L^*_{\e^{-1}} (\psi^{-1})^*, R^*_{\eta^{-1}} (\psi^{-1})^*, L_\e R_\eta \psi \big)
  \ \Big|\  \e,\eta \in A^\times, \psi\in\Aut A\Big\} .
$$
\end{lemma}

\begin{proof}
Let $\a,\b,\g\in\GL(A)$.
Suppose that
$\big( (\a^{-1})^*, (\b^{-1})^*, \g \big) \in \stab(w_A)$.
By Lemma~\ref{le:bilstab} we have
$\a(a)\b(b) = \g(ab)$ for all $a,b\in A$.
Plugging in~$1$ we get
$\a(a)\b(1) = \g(a)$ and $\a(1)\b(b) = \g(b)$.
Hence $\e:=\a(1)$ and $\eta:=\b(1)$ must be units of~$A$.
We define now
$\psi(a) := \e^{-1} \g(a)\eta^{-1}$.
Then we have $\psi(1)=1$ and
$$
\psi(a)\psi(b) =\e^{-1}\g(a)\eta^{-1}\, \e^{-1}\g(b) \eta^{-1}
 = \e^{-1} \a(a) \b(b)\eta^{-1} =\e^{-1} \g(ab) \eta^{-1} = \psi(ab).
$$
Therefore $\psi\in\Aut A$. By construction,
$\a= L_\e\psi$, $\b=R_\eta\psi$, and $\g=L_\e R_\eta\psi$,
and hence
$(\a^{-1})^* = L^*_{\e^{-1}} (\psi^{-1})^*$,
$(\b^{-1})^* = R^*_{\eta^{-1}} (\psi^{-1})^*$.
The argument is reversible.
\end{proof}

\subsection{Proof of Proposition~\ref{pro:stabut}}

Let $S_m$ denote the diagonal embedding of the symmetric group in~$G_m$.
Obviously, $\Tc_m\cap S_m =\{\Id\}$.
It is easy to see that $S_m$ normalizes $\Tc_m$.
Hence $\Tc_m S_m$ is a subgroup
of $G_m$ and $\Tc_m$ is a normal divisor of $\Tc_m S_m$.
It remains to prove that the stabilizer~$H_m$ equals $\Tc_m S_m$.
The inequality $\Tc_m S_m \subseteq H_m$ is obvious.

Note that $\langle m\rangle$ is the structural tensor of the algebra $A=\C^m$.
It is straightforward to check that
$\Aut A =\{P_\pi \mid \pi \in S_m \}$ .
Note that $(P_\pi^{-1})^* = P_\pi$ .
Hence Lemma~\ref{le:deGroote} implies
$$
 H_m = \stab(\langle m\rangle) = \big\{
 (\diag(\e^{-1}) P_\pi, \diag(\eta^{-1}) P_\pi,\diag(\e\eta) P_\pi) \mid
 \e,\eta \in (\C^\ti)^m, \pi\in S_m \big\}
$$
and we obtain $H_m = \Tc_m S_m$.
\qed

\begin{lemma}\label{le:ut-unique}
If the stabilizer of $w\in\C^m\ot\C^m\ot\C^m$ contains~$H_m$,
then $w=c\,\lan m\ran$ for some $c\in \C$.
\end{lemma}

\begin{proof}
Assume the stabilizer of
$w=\sum w_{ijk}e _i\otimes e_j\otimes e_k$
contains $H_m$.
By contradiction, we suppose that $w_{ijk}\ne 0$ for some
$i,j,k$ with $i\ne k$.
For any $(\diag(a),\diag(b),\diag(c))\in \Tc_m$ we have
$a_ib_jc_k w_{ijk} = w_{ijk}$ and hence
$a_ib_jc_k =1 =a_k b_k c_k$, which implies
$a_i = a_k b_k/b_j$.
However, defining
$\tilde{a}_i = 2 a_i$, $\tilde{c}_i = \frac12 c_i$,
$\tilde{a}_\ell = a_\ell,\ \tilde{c}_\ell = c_\ell$ 
for $\ell\ne i$ we get 
$(\diag(\tilde{a}),\diag(b),\diag(\tilde{c}))\in \Tc_m$.
This yields the contradiction
$\tilde{a}_i = \tilde{a}_k b_k/b_j = a_k b_k/b_j = a_i´$. 
We have thus shown that $w_{ijk}\ne 0$ implies $i=k$.
By symmetry, we conclude that $w_{ijk}=0$ unless $i=j=k$.
Finally,
from the invariance of $w$ under $S(m)$, we get $w_{iii} = w_{111}$ for all~$i$.
Hence $w=w_{111}\lan m\ran$.
\end{proof}

\subsection{Stability}

We need some criterion for testing stability.
By a {\em one-parameter subgroup} of~$G_s$ we understand
a morphism $\s\colon \C^\ti\to G_s$ of algebraic groups.
The {\em centralizer} $Z_{G_s}(R_s)$ of a subgroup $R_s$ of $G_s$
is defined as the set of $g\in G_s$ such that
$gh=hg$ for all $h\in R_s$.
For instance, let $T_s$ denote the maximal torus of $G_s$.
Then we have $Z_{G_s}(T_s) = T_s$.

The following important stability criterion
is a consequence Kempf's~\cite{kempf:78} refinement
of the Hilbert-Mumford criterion.

\begin{theorem}\label{th:stb-crit}
Let $w\in W$ be a tensor and $R_s$ be a reductive subgroup of $G_s$
contained in the stabilizer of $w$.
We assume that for all one-parameter subgroups~$\s$ of $G_s$,
with image in the centralizer $Z_{G_s}(R_s)$,
the limit $\lim_{t\to 0}\s(t)w$ lies in the $G_s$-orbit of $w$,
provided the limit exists. Then $w$ is stable.
\end{theorem}

\begin{proof}
Suppose that $w$ is not stable. Then there is a nonempty closed $G_s$-invariant
subset~$Y$ of $\ol{G_sw}\setminus G_sw$.
Kempf's result~\cite{kempf:78} states that there exists
a one-parameter subgroup
$\s\colon\C^\ti\to Z_{G_s}(R_s)$ such that
$\lim_{t\to 0}\s(t)w \in Y$.
Hence this limit does not lie in $G_sw$.
\end{proof}

\subsection{Proof of Proposition~\ref{th:unit-stable}}

We apply Theorem~\ref{th:stb-crit} with $R_s := \Tc_m\cap G_s$.
The group $R_s$ is a torus and hence reductive~\cite{kraf:84}.
It is easy to see that $Z_{G_s}(R_s)$ equals the maximal torus~$T_s$ of $G_s$.

Any one-parameter subgroup $\s\colon\C^\times\to T_s$ is of the form
$\s(t)=(\s_1(t),\s_2(t),\s_3(t))$ with
$$
 \s_1(t) = \diag(t^{\mu_{i}}),\
 \s_2(t) = \diag(t^{\nu_{i}}),\
 \s_3(t) = \diag(t^{\pi_{i}})
$$
with integers $\mu_i,\nu_i,\pi_i$.
Since $\det\s_k(t)= 1$ we must have
\begin{equation*}
 \sum_{i}\mu_{i} = 0,\
 \sum_{i}\nu_{i} = 0,\
 \sum_{i}\pi_{i} = 0.
\end{equation*}
We have
$\s(t)\lan m\ran =\sum_i t^{\mu_i + \nu_i +\pi_i} e_i \ot e_i \ot e_i$.
If the limit of $\s(t)\lan m\ran$ exists for $t\to 0$,
then $\mu_i + \nu_i +\pi_i \ge 0$ for all~$i$.
On the other hand,
$\sum_{i} (\mu_{i} + \nu_i +\pi_i) = 0$.
It follows that $\mu_{i} + \nu_i +\pi_i=0$ for all~$i$,
hence
$\s(t)\lan m\ran =\lan m\ran$.
\qed

\subsection{Proof of Lemma~\ref{le:unique-min}}

Recall that a partition $\la=(\la_1,\la_2,\ldots)$ is a weakly decreasing sequence
of natural numbers such that only finitely many components are nonzero.
We define the {\em length} of~$\la\ne 0$ as $\len(\la):=\max\{i \mid \la_i \ne 0\}$
and we put $\len(0):=0$.
The componentwise sum of partitions is well defined.
Note that $\la \preceq \mu$ implies  $\la+\nu \preceq \mu+\nu$
for any partitions $\la,\mu,\nu$.
We call $s(n) := (n,n-1,\ldots,1)$ the {\em symmetric staircase partition}
with $n$ rows and $n$ columns.
Note that if $\la$ is a regular partition with at least $j$ nonzero rows, then
$\la-s(j)$ is again a partition, since $\la$ has at least $j-i+1$ boxes in row $i$.

(1) Let $d=qm + r$ with $0\le r < m$. Then
$\square_m(d) :=  (q^m) + (1^r)$  (frequency notation)  
is the unique smallest element of $\Par_m(d)$.
The corresponding diagram has $q$ columns of length~$m$ plus
one additional column of length~$r$.

For given $d,m\in\N$ we set
$\ell := \ell(m,d) := \max\{n\leq m \mid n(n+1)/2 \leq d\}$
and we define the {\em staircase partition}
$\bot_m(d) := s(\ell) + \square_{\ell}(d-|s(\ell)|)$, see Figure~\ref{fig:bot}.
We observe that $\len(\bot_m(d))=\ell$ and moreover $\bot_m(d) = \bot_{\ell}(d)$.
\begin{figure}[h]
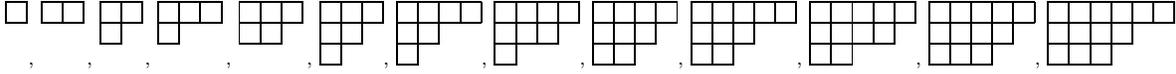

{\tiny
\parbox[b][0.875cm][t]{0.3cm}{\young(\ )},
\parbox[b][0.875cm][t]{0.6cm}{\young(\ \ )},
\parbox[b][0.875cm][t]{0.6cm}{\young(\ \ ,\ )},
\parbox[b][0.875cm][t]{0.9cm}{\young(\ \ \ ,\ )},
\parbox[b][0.875cm][t]{0.9cm}{\young(\ \ \ ,\ \ )},
\young(\ \ \ ,\ \ ,\ ),
\young(\ \ \ \ ,\ \ ,\ ), \young(\ \ \ \ ,\ \ \ ,\ ), \young(\ \ \ \ ,\ \ \ ,\ \ ),
\young(\ \ \ \ \ ,\ \ \ ,\ \ ), \young(\ \ \ \ \ ,\ \ \ \ ,\ \ ),
\young(\ \ \ \ \ ,\ \ \ \ ,\ \ \ ),
\young(\ \ \ \ \ \ ,\ \ \ \ ,\ \ \ )} 
\caption{The staircase partitions $\bot_3(d)$ for $d=1,\ldots,13$.}\label{fig:bot}
\end{figure}

Part~(1) of Lemma~\ref{le:unique-min} claims that
$\bot_m(d) \preceq\beta$ for any regular partition~$\beta\in\Par_m(d)$.

For showing this, put $\tilde{\ell}: =\len(\beta)$ and
note that $\beta - s(\tilde{\ell})$ is a partition by the previous observation.
If we had
$\ell +1 \le \tilde \ell$, then 
$\tilde{\beta} - s(\ell +1)$ would be a partition as well and hence
$d = |\beta| \geq \frac{(\ell+1)(\ell+2)}{2}$ contradicting the maximality of $\ell$.
So we have $\len(\bot_m(d)) =\ell \geq \tilde \ell$ and hence
$\bot_m(d) - s(\tilde \ell)$ is a partition.

We note that the subpartition consisting of the first $\tilde \ell$ rows of $\bot_m(d) - s(\tilde \ell)$
equals $\square_{\tilde \ell}(d')$ for some $d'\leq d-|s(\tilde \ell)|$.
Moreover,
$\square_{\tilde \ell}(d') \preceq \square_{\tilde \ell}(d-|s(\tilde \ell)|) \preceq \beta-s(\tilde \ell)$,
where the last inequality is due to the minimality of $\square_{\tilde{\ell}}(\cdot)$. 
It follows that
$\bot_m(d) - s(\tilde \ell) \preceq \beta-s(\tilde \ell)$,
which completes the proof of part~(1).

(2). 
For $k > \frac{m(m+1)}{2}+m$
we define the following regular partition in $\Par_{m+1}(d+km)$: 
$$
\la^k_{\reg} := (\la_1+k-1,\la_2+k-2,\ldots,\la_m+k-m,\frac{m(m+1)}{2}) .
$$ 
Part~(1) yields $\bot_{m+1}(d+km) \preceq \la^k_{\reg}$. 
Since $\la^k_{\reg} \preceq (\la_1+k,\ldots,\la_m+k,0)$ 
the claim follows.  
\qed

\section*{Appendix: Section~\ref{se:mamu}}

\subsection{Proof of Proposition~\ref{th:stab-mamu}}

We provide the proof in the cubic case only and thus assume $U_i=\C^n$.
The matrix multiplication tensor $\mamU$ is the
structural tensor of the associative algebra $A=\End(U)$.
Note that $A^\times=\GL(U)$.
Recall that $L_a,R_b\colon A\to A$ denote the left multiplication with $a$
and the right multiplication with $b$, respectively ($a,b\in A$).
If we interpret $A=U\ot U^*$, then we have $L_aR_b = a\ot b^*$.
Lemma~\ref{le:deGroote} states that any element~$g$ of
$\stab(\mamU)$ is of the form
$$
 g = \big( L^*_{\e^{-1}} (\psi^{-1})^*, R^*_{\eta^{-1}} (\psi^{-1})^*, L_\e R_\eta \psi \big)
$$
for some $\e,\eta \in A^\times, \psi\in\Aut A$.
The Skolem-Noether Theorem~\cite{hers:94} implies that any automorphism $\psi$ of $A$
is of the form $\psi=L_\rho R_{\rho^{-1}}$ for some $\rho\in A^\times$.
We thus obtain
$$
 \psi^{-1}L_{\e^{-1}} = L_{\rho^{-1}}R_\rho L_{\e^{-1}}
 = L_{\rho^{-1}\e^{-1}}R_\rho = \rho^{-1}\e^{-1} \ot \rho^*,
$$
which implies
$L^*_{\e^{-1}} (\psi^{-1})^* = (\psi^{-1} L_{\e^{-1}})^* = ((\e\rho)^{-1})^* \ot \rho$.
Similarly, we obtain
$R^*_{\eta^{-1}} (\psi^{-1})^* = (\rho^{-1})^* \ot \eta^{-1}\rho$.
Finally,
$L_\e R_\eta\psi = \e\rho \ot (\rho^{-1}\eta)^* \simeq (\rho^{-1}\eta)^* \ot \e\rho$.
(The flip $\simeq$ is due to our convention
$\Hom(U_3,U_1) \simeq U_3^*\ot U_1$,
unlike $\Hom(U_2,U_1) \simeq U_1\ot U_2^*$, $\Hom(U_3,U_2) \simeq U_2\ot U_3^*$.)
Setting $\a_1=\e\rho$, $\a_2=\rho$, $\a_3=\eta^{-1}\rho$ we see that~$g$
has the required form.
\qed

\begin{lemma}\label{pro:charstabmamu}
If the stabilizer of
$w\in  (U_1^*\ot U_2)\ot (U_2^*\ot U_3)\ot (U_3^*\ot U_1)$
contains the stabilizer~$\Hc$ of $\mamU$, 
then $w=c\,\mamU$ for some $c\in \C$. 
\end{lemma}

\begin{proof}
We have to show that the space of $\Hc$-invariants of
$(U_1^*\ot U_2)\ot (U_2^*\ot U_3)\ot (U_3^*\ot U_1)$ is one-dimensional.
Due to the description of $\Hc$ in Proposition~\ref{th:stab-mamu}
it suffices to prove that the space of $\GL(U_i)$-invariants
of $U_i^*\ot U_i$ is one-dimensional.
The latter follows as a special case of~\eqref{eq:LWRdual}
taking $\mu_i=(1,0,\ldots,0)$.
\end{proof}

\subsection{Proof of Proposition~\ref{th:mamu-stable}}

We follow \cite[Proposition 5.2.1]{meyer:06}.
Assume that $U_i=\C^{n_i}$.
Let $T(K_s)$ and $T_s$ denote the maximal tori of 
$\Sc_s := \SL(U_1) \ti\SL(U_2)\ti\SL(U_3)$ and $\Gc_s$,
respectively, consisting of triples of diagonal matrices with determinant~$1$.
It is clear that $R_s :=\Phi(T(K_s))$ is a subgroup of $T_s$.
Since $R_s$ is a connected subgroup of a torus, it is itself a torus and
thus reductive~\cite{kraf:84}.

We claim that $T_s$ equals the centralizer of $R_s$ in $\Gc_s$.
Indeed suppose that $g=(g_1,g_2,g_3)\in\Gc_s$ commutes with
all elements of $R_s$. Then $g_1$ commutes with all
diagonal matrices $\diag(a_i b_j^{-1})$, where
$a_1\cdots a_{n_1}=1$ and $b_1\cdots b_{n_2}=1$.
It is possible to choose $a_i,b_j$ such that $a_i b_j^{-1}$ are
pairwise distinct. Therefore $g_1$ must be a diagonal matrix.
Similarly, $g_2,g_3$ must be diagonal so that $g\in\Gc_s$.

We apply now Theorem~\ref{th:stb-crit} to the reductive subgroup
$R_s$ of the stabilizer $\Hc$ of $\mamU$.
Any one-parameter subgroup $\s\colon\C^\times\to T_s$ is of the form
$\s(t)=(\s_1(t),\s_2(t),\s_3(t))$ with
$$
 \s_1(t) = \diag(t^{\mu_{ij}}),\
 \s_2(t) = \diag(t^{\nu_{jk}}),\
 \s_3(t) = \diag(t^{\pi_{ki}}) ,
$$
where $\mu_{ij},\nu_{jk},\pi_{ki}\in\Z$ for $i\le n_1,j\le n_2, k\le n_3$.
Since $\det\s_1(t)= \det\s_2(t) = \det\s_3(t) = 1$ we must have
\begin{equation}\label{eq:sum0}
 \sum_{i,j}\mu_{ij} = 0,\
 \sum_{j,k}\nu_{jk} = 0,\
 \sum_{k,i}\pi_{ki} = 0.
\end{equation}
Let $(e_{ij}), (e_{jk}), (e_{ki})$ denote the
standard bases of
$\C^{n_1\ti n_2},\C^{n_2\ti n_3},\C^{n_3\ti n_1}$,
respectively. The matrix multiplication tensor
can then be expressed as
$$
 \lan n_1,n_2,n_3\ran = \sum_{i,j,k} e_{ij} \ot e_{jk} \ot e_{ki} .
$$
We have
$$
 \s(t) \lan n_1,n_2,n_3\ran =
  \sum_{i,j,k} t^{\mu_{ij} + \nu_{jk} + \pi_{ki}} e_{ij} \ot e_{jk} \ot e_{ki} .
$$
Suppose that the limit  of $\s(t) \lan n_1,n_2,n_3\ran$ for $t\to 0$ exist.
Then
$$
 \forall i,j,k\quad \mu_{ij} + \nu_{jk} + \pi_{ki} \ge 0 .
$$
Summing over all $i,j,k$ and using~\eqref{eq:sum0} we get
$$
 \sum_{i,j,k} (\mu_{ij} + \nu_{jk} + \pi_{ki}) =
 \sum_k \sum_{i j} \mu_{i j} + \sum_i \sum_{j,k}\nu_{jk} + \sum_j \sum_{k,i}\pi_{ki} = 0.
$$
Therefore, we have
$\mu_{ij} + \nu_{jk} + \pi_{ki} = 0$
for all $i,j,k$.
We conclude  that
$\lim_{t\to 0}\s(t) \lan n_1,n_2,n_3\ran = \lan n_1,n_2,n_3\ran$.
Theorem~\ref{th:stb-crit} implies that the $\Gc_s$-orbit of
$\lan n_1,n_2,n_3\ran$ is closed.
\qed


\section*{Appendix: Section~\ref{se:comp}}

\subsection{Explicit realizations of Schur-Weyl modules}\label{se:explicit-real}

For the following well known facts see~\cite{fult:97,fuha:91}. 
Let $V=\C^m$ with the standard basis $e_1,\ldots,e_m$. 
For a partition $\la\vdash_m d$ we denote by $\Tau_m(\la)$ 
the set of tableaux~$T$ of shape~$\la$
with entries in $\{1,2,\ldots,m\}$. 
Every $T\in\Tau_m(\la)$ has a content~$\a\in\N^m$, 
where $\a_j$ counts the number of occurences of $j$ in $T$. 

Let $\St_\la$ denote the standard tableau arising when we 
number the boxes of the Young diagram of $\la$ 
columnwise downwards, starting with the leftmost column,
cf.~Figure~\ref{fig:Sla}.
We assign to $T\in\Tau_m(\la)$ the basis vector
$e(T):=e_{j_1}\ot \cdots \ot e_{j_d}\in V^{\ot d}$, where 
$j_k\in\{1,\ldots m\}$ is the entry of $T$ at the box 
which is numbered~$k$ in $\St_\la$.
In other words, $j_k$ is the $k$th entry of~$T$
when we read the tableau~$T$ 
columnwise downwards, starting with the leftmost column.
Note that $e(T)$ is a weight vector with respect to the  
subgroup $T_m\subseteq \GL_m$ of diagonal matrices, 
and the weight of $e(T)$ equals the content~$\a$ of $T$. 
One should think of the tableau $T$ as a convenient way 
to record the basis vector $e(T)$.
\begin{figure}[h]
\young(1468,257,3)
\caption{\small $\St_\la$ for $\la=(4,3,1)$,
$P_\la=\Perm\{1,4,6,8\}\ti \Perm\{2,5,7\}\ti\Perm\{3\}$,
$Q_\la=\Perm\{1,2,3\}\ti\Perm\{4,5\}\ti\Perm\{6,7\}\ti\Perm\{8\}$.
}\label{fig:Sla}
\end{figure}

Let $P_\la$ the subgroup of permutations in $\St_d$ preserving
the rows of $\St_\la$ and denote by 
$Q_\la$ the subgroup of permutations in $\St_d$ preserving
the columns of $S_\la$, cf.~Figure~\ref{fig:Sla}.
To any $T\in\Tau_m(\la)$ we assing now the following vector in~$V^{\ot d}$: 
$$
 v(T) := \frac1{|P_\la| |Q_\la|} \sum_{\s\in Q_\la \atop \pi\in P_\la}
              \sgn(\s)\, \s \pi\, e(T) . 
$$
Note that $v(T)$ is a weight vector and its weight equals the content~$\a$ of $T$. 

Let $T_\la$ denote the semistandard tableau which in the 
$i$th row only has the entry $i$. Clearly, $T_\la$ has the content~$\la$. 
Let $\mu = \la'$ denote the partition dual to $\la$ and let $\ell$ be the length of~$\mu$. 
Then we have 
\begin{equation}\label{eq:defvla}
 v_\la := v(T_\la) = (e_1\wedge\ldots\wedge e_{\mu_1}) \ot \cdots\ot 
                                  (e_{1} \wedge\ldots \wedge e_{\mu_{\ell}}) .
\end{equation}
It is easy to see that $v_\la$ is a $U_m$-invariant weight vector of weight $\la$, 
where $U_m\subseteq \GL_m$ denotes the subgroup of upper triangular matrices 
with ones on the main diagonal. 
Hence the $\GL_m$-submodule $V_\la$ generated by $v_\la$ is 
irreducible of highest weight $\la$. 
We have
$$
 V_\la =\span\big\{  (x_1\wedge\ldots\wedge x_{\mu_1}) \ot \cdots\ot 
                                  (x_{1} \wedge\ldots \wedge x_{\mu_{\ell}}) \mid x_i \in V \big\}
$$ 
and $V_\la$ is spanned by 
$\{v(T) \mid T\in \Tau_m(\la) \}$.
It is well known that the $v(T)$ form a basis of~$V_\la$ when $T$ 
runs over all semistandard tableaux in $\Tau_m(\la)$.
A basis of the weight space $V_\la^\a$ is provided by the~$v(T)$ 
where $T\in \Tau_m(\la)$ runs over all semistandard tableaux 
with content~$\a$. 
We embed $S_m$ into $\GL_m$ by mapping $\pi\in S_m$ to the 
permutation matrix $P_\pi$. 
Note that the group~$S_m$ acts on $\Tau_m(\la)$ by permutation of the entries 
of the tableaux. 
Then we have 
$P_\pi v(T) = v(\pi T)$ for $T\in\Tau_m(\la)$. 
We have thus found explicit realizations 
of Schur-Weyl modules.

Suppose now $d=m$ and consider 
the weight space~$\Sp_\la :=V_\la^{\e_m}$ 
of weight $\e_m=(1,\ldots,1)$. 
Then $\Sp_\la$ is a $S_m$-submodule of~$V_\la$. 
Moreover, the vectors $v(T)$, where 
$T$ runs over all standard tableaux of shape~$\la$, 
provide a basis of $\Sp_\la$. 
One can show that $\Sp_\la$ is irreducible and isomorphic to $[\la]$. 
We have thus also found an explicit realization of the 
irreducible $S_m$-module $[\la]$.

\subsection{Tableaux straightening}

An explicit description of the action 
the subgroup $S_m$ of $\GL_m$ on the basis $(v(T))$
is provided by the following tableau straightening algorithm~\cite[p.97-99, p.110]{fult:97}.  
It takes as input any tableau~$T\in\Tau_m(\la)$ and expresses the vector~$v(T)$ 
as an integer linear combination of the basis vectors $v(S)$, 
where $S$ is semistandard.  
This way, we obtain an explicit description of 
the operation of $\stab(\a)$ on the weight space $V_\la^\a$, which is 
required for applying Theorem~\ref{th:ut-dimVH}.

\begin{enumerate}
 \item If $T$ is semistandard, return $v(T)$. 
 \item If the columns of $T$ do no have pairwise distinct entries, return~$0$.
       Otherwise, apply column permutations~$\pi$ to put all columns in strictly increasing order 
       by applying the rule $\pi v(T) = \sgn(\pi)\, v(\pi T)$. 
 \item If the resulting tableau is not semistandard, suppose the $k$th entry of the $j$th column
       is strictly larger than the $k$th entry of the $(j+1)$th column.
       Then we have $v(T) = \sum_S v(S)$, where $S$ ranges over all tableaux that arise from~$T$ 
       by exchanging the top $k$ elements from the $(j+1)$th column with any selection of 
       $k$~elements in the $j$th column, preserving their vertical order. 
       Continue recursively with the resulting~$S$.
\end{enumerate}
See \cite[p. 110]{fult:97} for a proof that this algorithm terminates
(whatever choice of $k$ and $j$ is made in step~(3)).

\subsection{Explicit construction of highest weight vectors in $\Oh(W)$}
\label{se:hwv}

Our goal here is to give an explicit description of the space of highest weight vectors 
of $\Oh(W^*)_d$ that is amenable to calculations 
(we replaced $W$ by $W^*$ to simplify notation). 
We shall proceed in several steps. 

1. Recall from~\eqref{eq:defvla} that $v_\la$ is a highest weight vector of $V_\la$. 
We have (recall~\S\ref{se:HWV}) 
\begin{equation}\label{eq:hwv}
 \HWV_\la (V^{\ot d}) = \span\big\{ \pi\, v_\la \mid \pi \in S_d \big\} .
\end{equation}
This follows from Schur-Weyl duality~\eqref{eq:SW-duality} 
and the fact that $[\la]$ is spanned by
the $S_d$-orbit of any of its nonzero vectors.

2. Now assume $W_i=\C^{m_i}$,  
consider $W:= W_1 \otimes W_2 \otimes W_3$ of format $\um$
with the group $G=\GL(W_1)\ti\GL(W_2)\ti\GL(W_3)$ acting on $W$ and let 
$\underline \la = (\la_1,\la_2,\la_3)\in\Lambda^+(\um)$. 
We define 
$v_\ula :=  v_{\la_1} \ot v_{\la_2} \ot v_{\la_3} \in W^{\ot d}$. 
Note that the group $(S_d)^3$ acts on $W^{\ot d}$.
Equation~\eqref{eq:hwv} implies 
$$
\HWV_\ula (W^{\ot d}) 
 = \HWV_{\la_1}(W_1^{\ot d})\ot\HWV_{\la_2}(W_2^{\ot d})\ot\HWV_{\la_3}(W_3^{\ot d}) 
 = \span\big\{\underline{\pi}\, v_\ula \mid \underline{\pi} \in (S_d)^3 \big\} .
$$ 

3. We embed $S_d$ into $(S_d)^3$ diagonally via $\pi \mapsto (\pi,\pi,\pi)$.
The induced action on $W^{\ot d}$ just permutes the different copies of $W$.
Consider now the projection 
$P_\Sym\colon W^{\ot d} \to \Sym^d W$ 
onto the subspace of symmetric tensors
given by ${d!}^{-1} \sum_{\pi \in S_d} \pi$.
Recall also $\Oh(W^*)_d\simeq \Sym^d W$. 
We arrive at the desired characterization
\begin{equation*}\label{le:hwvO}
 \HWV_\la(\Oh(W^*)_d) = \HWV_\la(\Sym^d W) =
 P_\Sym\big(\HWV_\la(W^{\ot d})\big) = 
  \span \big\{P_\Sym (\underline{\pi}\, v_{\ula}) \mid \underline \pi \in S_d^3 \big\} .
\end{equation*}

Let $w\in W^*$. In order to prove that $V_\ula(G)$ occurs in
$\Oh(\overline{G w})$ 
it is sufficient to exhibit some 
$\underline{\pi} \in (S_d)^3$ and some $g \in G$ such that
$P_\Sym (\underline{\pi}\, v_{\underline{\la}}) (gw) \neq 0$ 
as polynomial evaluation.
A straightforward algorithm for this evaluation requires at least 
$R(w)^d$ steps and thus only allows the study of small examples 
in practice. 

\subsection{Details of the proof of Lemma~\ref{le:G-obst}}

Recall from \S\ref{se:rep-unit} that the weight space $V_\la^\a$ is invariant 
under the action of $\stab(\a)$. We are interested in the splitting of $V_\la^\a$ 
into irreducible $\stab(\a)$-modules. 

\begin{remark}\label{re:gay}
In the special case 
$\a=d\e_m$, where $\stab(\a)=S_m$, it is known~\cite {gay:76} that 
the arising multiplicities are special plethysm coefficients, namely  
$$
\mult([\pi],V_\la^{d\e_m}) = \mult(V_\la(\GL_m), S_\pi(\Sym^d \C^m))\ 
\mbox{ for $\pi\vdash m$}.
$$ 
\end{remark}

\begin{lemma}\label{le:split1}
Let $0\le s < m$, $d\ge 1$. Then 
$V_{(dm-s) 1^s}^{d\e_m} \simeq [(m-s)\,1^s]$ as $S_m$-modules. 
\end{lemma}

\begin{proof}
Let $\mathcal{ST}$ denote the set of semistandard tableaux of shape $(dm-s) 1^s$ and content $d\e_m$. 
Moreover, let $\mathcal{S}$ denote the set of standard tableaux of shape $(m-s) 1^s$.
Let $T\in\mathcal{ST}$ and suppose that $1,a_1,\ldots,a_s$ are the entries of the first column of $T$. 
After deleting $d-1$ of the boxes with the entries $1,\ldots,m$ from the first row of $T$,
we obtain a standard tableau $\psi(T)\in\mathcal{S}$.
It is clear that $\psi\colon\mathcal{ST}\to\mathcal{S}$ is a bijection. 
The algorithmic description of the Schur-Weyl modules above easily implies that 
$\psi(\pi T)= \pi \psi(T)$ for any $\pi\in S_m$.
We use now that $(v(T))_{T\in\mathcal{ST}}$ and 
$(v(T'))_{T'\in\mathcal{S}}$ form a basis of 
the weight space $V_{(dm-s) 1^s}^{d\e_m}$ and 
the $S_m$-module $[(m-s)\,1^s]$ realized as a submodule of 
$(\C^m)^{\ot m}$ as in \S\ref{se:explicit-real}.
\end{proof}

Taking $d=2$ and $s=3$ we get  
$\big(V_{2\e_m}^{2\e_m} \ot V_{2\e_m}^{2\e_m} \ot V_{(2m-3) 1^3}^{2\e_m}\big)
 \simeq \big([m] \ot [m] \ot [(m-3)\, 1^3] \big)^{S_m} = 0$.
 
\begin{lemma}\label{le:split2}
As $S_{m-1}\ti S_2$-modules we have 

1. $V_{(2^m)}^{2^{m-1}1^2} \simeq [m-1]\ot [2]$.  

2. $V_{(2m-3) 1^3}^{2^{m-1} 1^2} \simeq 
  \big([(m-4)\, 1^3]\ot [2]\big) \oplus \big([(m-2)\, 1]\ot [1^2]\big) \oplus 
   \big([(m-3)\, 1^2]\ot [2]\big) \oplus \big([(m-3)\, 1^2]\ot [1^2]\big)$.
\end{lemma}

\begin{proof}
1. There is a single semistandard tableau~$T$ of shape $2^m$ and content $2^{m-1}\,1^2$: 
the $i$th row contains the entries $i,i$ for $i<m$ and the $m$th row contains $m,m+1$. 
The tableau $T$ is fixed by the action of $S_{m-1}$. 
The transposition~$\pi$ in $S_2$ exchanges $m$ and $m+1$. 
However, the straightening algorithm shows that $v(\pi T)=v(T)$. 

2. The basis of $V_{(2m-3) 1^3}^{2^{m-1} 1^2}$ is indexed by semistandard tableaux 
which fall into four different classes as indicated in Figure~\ref{fig:C12} 
and Figure~\ref{fig:C3}.
\newcommand{\tabNmE}{{\text{\tiny $m\!\!\!-\!\!\!1$}}}
\newcommand{\tabNpE}{{\text{\tiny $m\!\!\!+\!\!\!1$}}}

\begin{figure}[h]
$$
 \young(112233\cdots\tabNmE\tabNmE m\tabNpE,a,b,c)\quad
 \young(112233\cdots\tabNmE\tabNmE,a,m,\tabNpE)\
$$
\caption{Tableaux of class 1 and class 2.}\label{fig:C12}
\end{figure}

\begin{figure}[h]
$$
  \young(112233\cdots\tabNmE\tabNmE \tabNpE,a,b,m)\quad
  \young(112233\cdots\tabNmE\tabNmE m,a,b,\tabNpE)
$$
\caption{Tableaux of class~3 and class~4.}\label{fig:C3}
\end{figure}

Omitting the boxes with entries $m,m+1$ in the tableaux of class~1 
and deleting repeated entries in the first row, we obtain 
a bijection of the set of tableaux of class~1
with the set of standard tableaux of shape $(m-4)\,1^3$. 
It follows that the span of the basis vectors of class~1 is isomorphic to 
$[(m-4)\, 1^3]\ot [2]$. 

Similarly, the tableaux of class~2 are in bijection with the standard tableaux 
of the shape $(m-2)\,1^3$.
The span of the basis vectors of class~2 is isomorphic to
$[(m-2)\,1^3]\ot [1^2]$ (note the sign change when permuting $m$ with $m+1$). 

Let $\mathcal{T}$ denote the set of tableaux of class~3 and consider the transposition
$\pi:=(m\, m+1)$. 
Then $\{\pi T \mid T\in \mathcal{T}\}$ is the set of tableaux of class~4. 
Clearly, $\mathcal{T}$ is in bijection with the standard tableaux of shape $(m-3)\, 1^2$. 
The vectors $v(T) + \pi v(T)$ for $T\in\mathcal{T}$
span $[(m-3)\,1^2]\ot [2]$, whereas the vectors $v(T) - \pi v(T)$ 
span $[(m-3)\,1^2]\ot [1^2]$. 
\end{proof}

Lemma~\ref{le:split2} implies now 
$$
\big(V_{(2^m)}^{2^{m-1}1^2} \ot V_{(2^m)}^{2^{m-1}1^2} 
   \ot V_{(2m-3) 1^3}^{2^{m-1} 1^2}\big)^{S_{m-1}\ti S_2} = 0,
$$
which was needed in the proof of the first part of Lemma~\ref{le:G-obst}. 
We omit the details of the proof of the second part of this lemma,
see however~\S\ref{se:hwv}.

The following was claimed in Remark~\ref{re:mu}. 

\begin{lemma}
Let $n\ge 2$. Then  
$g(2^{n^2},(2n)^n,(2n)^n) = 1$ and  
$g((2n^2-3)\,1^3,(2n)^n, (2n)^n) > 0$.
\end{lemma}

\begin{proof}
The first claim follows from~\cite[Satz~3.1]{clme:93}. 
For the second claim, we note that 
$\ula:= ((2n^2-3)\,1^3,(2n)^n, (2n)^n)$ 
can be decomposed as 
$\ula= \umu + (n-1)\cdot (2n,2^n,2^n)$,
where $\umu:= ((2n-3)\, 1^3, 2^n, 2^n)$. 
It is clear that $g(2n,2^n,2^n)=1$. 
It follows from~\cite{rewe:94, rosa:01} that 
$g(\umu) = g((2n-3,1^3),(n^2),(n^2)) = 1$.
Since the triples with positive Kronecker coefficients form a semigroup,
the second assertion follows. 
\end{proof}

\section*{Appendix: Section~\ref{se:exten}}

\begin{proposition}\label{pro:norm-ext}
Let $Z$ be an irreducible normal complex algebraic variety
and $f$ be a regular function defined on a nonempty Zariski open subset of $Z$.
If $f$ has an extension to $Z$ which is continuous in the $\C$-topology, 
then $f$ has a regular extension to $Z$.
\end{proposition}

\begin{proof}
The proof uses some standard facts from algebraic geometry~\cite[III,\S8]{mumf:88}. 
Since $Z$ is normal, the rational function~$f$ has a well defined divisor. 
Suppose that $f$ had a pole of multiplicity $k\ge 1$ at the irreducible hypersurface $H$ of $Z$. 
Let $p$ denote the vanishing ideal of~$H$.  
Then we can write
$f= p^{-k} g/h$ for some $g,h\in \Oh(Z)$ such that $g,h\not\in p$. 
Choose $z\in H$ such that $g(z)h(z)\ne 0$ and let $z_k$ be any sequence in $Z$ converging to $z$. 
Then $\lim_{k\to\infty} f(z_k) = \infty$, contradicting the assumption that $f$ has 
a $\C$-continuous extension to $Z$. 

Therefore, $f$ has no pole divisor and hence $f$ is a regular function.  
\end{proof}

\begin{proposition}\label{pro:affine}
Let $Z$ be an irreducible affine complex algebraic variety 
and $U\subseteq Z$ be a nonempty Zariski open subset, $U\ne Z$. 
If $U$ is affine, then $Z\setminus U$ is of pure codimension one in~$Z$.
\end{proposition}

\begin{proof}(Sketch) 
We use some standard facts from algebraic geometry. 
Let $\varphi\colon\tilde{Z} \to Z$ be the normalization of $Z$, cf.~\cite{hart:77}.
Then $\tilde{Z}$ and $\tilde{U}:=\varphi^{-1}(U)$ are affine. 
Let $C_i$ be the irreducible components of $Z\setminus U$. 
Since $\varphi$ is finite,  $C_i' :=\varphi^{-1}(C_i)$ are the irreducible 
components of $\tilde{Z}\setminus\tilde{U}$  and 
$\dim C_i'=\dim C_i$. 
We may therefore assume that $Z$ is normal.

Suppose first that all the irreducible components of $Z\setminus U$ have 
codimension at least two. 
Consider the injection $\iota\colon U\hookrightarrow Z$. 
The restriction morphism $\iota^*\colon\Oh(Z)\to\Oh(U)$ is bijective 
since, by a reasoning as in the proof of Proposition~\ref{pro:norm-ext}, 
any regular function on~$U$ can be extended to $Z$.
Since $Z$ and $U$ are assumed to be affine, $\iota$ is an isomorphism 
and hence $U=Z$. 
We omit the proof of the general case.
\end{proof}

The following general result is due to~\cite{mats:60}.

\begin{proposition}\label{pro:aff-red}
Let $G$ be a reductive group and $H$ be a closed subgroup.
Then $G/H$ is affine iff $H$ is reductive.
\end{proposition}

\subsection{Proof of Theorem~\ref{th:extend}}

The following lemma settles part~(1) of Theorem~\ref{th:extend}.

\begin{lemma}\label{le:cufo}
Suppose that $w\in W$ and $\e_\um \in S^o(w)$,
where $\um=(\dim W_1,\dim W_2,\dim W_3)$.
Then we have $m_1=m_2=m_3$.
\end{lemma}

\begin{proof}
We have $\det g=1 $ for all $g\in\stab(w)$
by our assumption $\e_\um \in S^o(w)$,
On  the other hand, $(a\,\Id_{m_1},b\,\Id_{m_2},c\,\Id_{m_3})\in \stab(w)$
for any $a,b,c\in\C^\times$ with $abc=1$. This implies
$$
  1 = \det(a\,\Id_{m_1},b\,\Id_{m_2},c\,\Id_{m_3}) = a^{m_1} b^{m_2} c^{m_3}
 = a^{m_1-m_3}\, b^{m_2 -m_3} .
$$
Therefore, $m_1=m_2=m_3$.
\end{proof}

The next lemma shows part~(2) of Theorem~\ref{th:extend}.

\begin{lemma}\label{le:limdet}
Let $w\in W\setminus\{0\}$ be stable and $u\in\overline{Gw}\setminus Gw$.
Suppose that $(g_n)$ is a sequence in $G$ such that
$\lim_{n\to\infty} g_n w = u$.
Then we have $\lim_{n\to\infty} \det g_n = 0$.
\end{lemma}

\begin{proof}
Since $G_s w$ is closed and $0\not\in G_s w$
we have
$$
 \e := \inf \{ \|\tilde{g} w\| \mid \tilde{g} \in G_s \}
     = \min \{ \|\tilde{g} w\| \mid \tilde{g} \in G_s \} >0 .
$$
For each~$n$ there are $\tilde{g}_n\in G_s$ such that
\begin{equation}\label{eq:gag}
 g_n w = \det g_n\ \tilde{g}_n w .
\end{equation}
Hence $\|g_n w\| = |\det g_n|\, \| \tilde{g}_n w\|$.
Since $\lim_{n\to\infty} \|g_n w\| = \|u\|$ and
$\|\tilde{g}_n w\| \ge \e >0$ we conclude that
$|\det g_n | \le \|g_n w\| /\e$ is bounded.

If $\lim_{n\to\infty} \det g_n = 0$ were false, then
there would  be some nonzero limit point $\d$
of the sequence $(\det g_n)$. After going over
to a subsequence, we have
$\lim_{n\to\infty} \det g_n = \d$.
From~\eqref{eq:gag} we get
$\lim_{n\to\infty} \tilde{g}_n w = \d^{-1} u$.
Hence $\d^{-1}u \in \overline{G_s u} = G_s u$,
which implies the contradiction $u\in Gw$.
\end{proof}

For the third part of Theorem~\ref{th:extend}
we note that for $w\in W$ and $g,h\in G$
\begin{equation}\label{eq:gdet}
 g\,\det_w(hw) = \det_w(g^{-1}hw) = \det(g^{-1}h)
 = \det(g)^{-1} \det(h) = \det(g)^{-1} \det_w(hw) .
\end{equation}
If $\det_w$ had a regular extension to~$\overline{Gw}$, then
\eqref{eq:gdet} shows that $\C\det_w$ is a submodule of
$\Oh(\overline{Gw})$ of highest weight $-\e_\um$. Hence
$\Oh(W)$ would contain an irreducible submodule of highest weight~$-\e_\um$
as well. On the other hand, the Kronecker coefficient $g(\e_\um)$ 
vanishes if $m>1$.
This contradicts~\eqref{eq:multOW}
and proves that $\det_w$ is not a regular function on $\overline{Gw}$.

Proposition~\ref{pro:norm-ext} combined with part~(2) and part~(3)
implies now that $\overline{Gw}$ is not a normal variety,
showing the fourth part of Theorem~\ref{th:extend}.

Part~(5) of Theorem~\ref{th:extend} follows by tracing the
proof of Proposition~\ref{th:stabSs}:
Let $f\in \Oh(Gw)_d$ be a highest weight vector.
The restriction $\tilde{f}$ of~$f$ to $G_sw$ does not vanish since
$Gw$ is the cone generated by $G_sw$.
So $\tilde{f}$ is a highest weight vector for the action of $G_s$.
Let $M$ denote the irreducible $G_s$-module generated by~$\tilde{f}$.
The $G_s$-equivariant restriction morphism
$\mathrm{res}\colon\Oh(\overline{Gw})\to \Oh(G_s w)$
is surjective since $G_sw$ is assumed to be closed.
Hence there exists an irreducible $G_s$-module~$N\subseteq \Oh(\overline{Gw})$
that maps to $M$ under $\mathrm{res}$.
We have $N\subseteq \Oh(\overline{Gw})_\d$ for some degree~$\d$.
Let $F$~be a highest weight vector of the $G_s$-module~$N$.
Then
$\mathrm{res}(F) = c \tilde{f}$ for some $c\in\C^\times$.
W.l.o.g.\ $c=1$.

For each $g\in G$ there exists $\tilde{g}\in G_s$ such that
$gw= \det g\, \tilde{g} w$. Moreover, since $F$ is homogeneous of degree~$\d$,
$$
 F(gw) = (\det g)^\d F(\tilde{g} w) = (\det g)^\d f(\tilde{g}w) .
$$
Moreover, since $f$ is homogeneous of degree~$d$,
$f(gw)= (\det g)^d f(\tilde{g}w)$.
We conclude that 
$$
 F(gw) = (\det g)^{\d -d} f(gw) = (\det_w (gw))^{\d -d} f(gw).
$$
Therefore, $(\det)^{\d -d} f = F$
is regular on $\overline{Gw}$ and the assertion follows.
\qed

\subsection{Proof of Corollary~\ref{cor:nn-mamu}}

(1) The first assertion is immediate from Theorem~\ref{th:extend}.

(2) Put $W=\C^m\ot\C^m\ot\C^m$, $w :=\langle m \rangle$.
Proposition~\ref{pro:stabut} implies that
$(\det g)^2 = 1$ for all $g\in\stab(w)$.
As in the proof of Theorem~\ref{th:extend} we show that
$\det_w^2\colon Gw \to\C$ has a continuous extension to~$\overline{Gw}$.

If $\det_w^2$ had a regular extension to~$\overline{Gw}$, then
$\Oh(W)$ would contain an irreducible submodule of highest weight
$\ula=(2^m,2^m,2^m)$ (compare~\eqref{eq:gdet}).
On the other hand, using the symmetry property
$g(\la,\mu,\nu)=g(\la',\mu',\nu)$
of Kronecker coefficients~\cite{StanleyVol2}
(with $\la'$ denoting the transposed partition),
we obtain
$g(2^m,2^m,2^m)=g(m^2,m^2,2^m)=0$ for $m\ge 5$.
(The vanishing since the right hand partition has more than four rows.)
This contradicts~\eqref{eq:multOW}
and proves that $\det_w^2$ does not have a regular extension to~$\overline{Gw}$.
The assertion follows now with Proposition~\ref{pro:norm-ext}.
\qed

\subsection{Proof of Proposition~\ref{pro:codim1}}

Put $H:=\stab(w)$ and $H_s:=H\cap G_s$.
If $w$ is stable, then $G_s/H_s \simeq G_sw$ is closed and hence affine.
By Proposition~\ref{pro:aff-red}, $H_s$ is reductive.
Consider the morphism
$H\to (\C^\times)^3, g\mapsto (\det g_1,\det g_2,\det g_3)$
with kernel $H_s$. Since reductiveness is preserved under
extensions, and closed subgroups of $(\C^\times)^3$
are reductive, it follows that $H$ is reductive~\cite{hump:87}.
Proposition~\ref{pro:aff-red} implies that $ G/H \simeq Gw$ is affine.
The last assertion follows by applying Proposition~\ref{pro:affine} 
to $Z=\overline{Gw}$ and $U=Gw$. 
\qed

\section*{Appendix: Section~\ref{se:mp}}

The easy proof of the following observation is left to the reader.

\begin{lemma}\label{le:fingen}
Let $\ula^1,\ldots,\ula^s \in S(w)$ be a system of generators for the
semigroup $S(w)$ with $\ula^i \in \Lambda^+_{d_i}(\um)$. Then $P(w)$
is the convex hull of
$\frac1{d_1}\ula^1,\ldots,\frac1{d_s}\ula^s$.
Moreover, any rational point of $P(w)$ is a rational convex combination
of these points.
\end{lemma}

\begin{lemma}\label{le:cenpo}
If $u_\um \in P(w)$ then there exists $\ell\ge 1$ such that $\ell\e_\um \in S(w)$.
\end{lemma}

\begin{proof}
By Lemma~\ref{le:fingen}, $u_\um$ is a rational convex combination of
$\frac1{d_1}\ula^1,\ldots,\frac1{d_s}\ula^s$. Hence there exist integers~$N_i\ge 0$,
$N>0$ such that
$Nu_\um = \sum_i N_i \ula^i \in S(w)$.
Moreover, $Nu_\um = \ell \e_\um$ where $\ell := N/m\in\N$.
\end{proof}

\subsection{Proof of Theorem~\ref{th:mpzero}}

For the second statement
take any $\ula=(\la_{12},\la_{23},\la_{31})$
in $\Lambda^+_{dn}(n^2,n^2,n^2)$.
Consider the rectangular partition
$(d^n)=(d,\ldots,d)\vdash_n dn$.
The main result in~\cite{buci:09} states that for
$ij= 12,23,31$
there exists a positive stretching factor $k_{ij}\in\N$ such that
$g(k_{ij}\la_{ij},(k_{ij}d)^n,(k_{ij}d)^n)\ne 0$.
Let $k$ be the least common multiple of $k_{12},k_{23},k_{31}$.
Then we have for $ij= 12,23,31$
$$
 g(k\la_{ij},(kd)^n,(kd)^n)\ne 0 .
$$
Theorem~\ref{cor:H-inv-mamu}
with $\mu_i = (kd)^n\vdash_n kdn$ implies that
$k\ula\in S^o(\lan n,n,n\ran)$.
Hence $\frac1{dn}\ula\in P^o(\lan n,n,n\ran)$.
Since the set of $\frac1{dn}\ula$ is dense in $\Delta(n^2,n^2,n^2)$,
we obtain $P^o(\langle n,n,n\rangle)=\Delta(n^2,n^2,n^2)$ as claimed.

The statement for the unit tensors is an easy consequence of
Corollary~\ref{cor:regmet}(2).
\qed

\subsection{Proof of Lemma~\ref{le:umint}}

Let $\ux=\frac1{d}\ula$ with $\ula\in S^o(w)\cap\Lambda^+_d(\um)$.
Proposition~\ref{th:stabSs} implies that there exists
$k\in\Z$ such that $\ula + k\e_\um \in S(w)$, cf.~\eqref{eq:SS0}.
We may assume $k>0$ due to Lemma~\ref{le:cenpo} and
our assumption $u_\um\in P(w)$.
Hence the point 
$$
 \frac1{d+km}(\ula + k\e_\um) = \frac{d}{d+km}\,\ux + \frac{km}{d+km}\,u_\um .
$$
lies in $P(w)$. From the convexity of $P(w)$, we conclude that
$\{t\ux + (1-t)u_\um \mid 0 \le t \le \d \} \subseteq P(w)$,
where $\d:=d/(d+km)$. Replacing $\d$ by the minimum of these values
for all generators of $S^o(w)$, we obtain that
$$
 \forall\, \ux \in P^o(w)\ \forall\, 0 \le t \le \d\
  t\ux + (1-t) u_\um \in P(w) .
$$
\qed

\end{document}

%% file: macros.tex
\newtheorem{theorem}{Theorem}[section]
\newtheorem{proposition}[theorem]{Proposition}
\newtheorem{lemma}[theorem]{Lemma}
\newtheorem{corollary}[theorem]{Corollary}
\theoremstyle{definition}
\newtheorem{definition}[theorem]{Definition}

\newtheorem{problem}[theorem]{Problem}

\theoremstyle{remark}
\newtheorem{remark}[theorem]{Remark}

\numberwithin{equation}{section}
\numberwithin{theorem}{section}

\def\N{\mathbb{N}}
\def\Z{\mathbb{Z}}

\def\R{\mathbb{R}}
\def\C{\mathbb{C}}

\def\la{\lambda}
\def\a{\alpha}
\def\b{\beta}
\def\g{\gamma}
\def\s{\sigma}
\def\d{\delta}
\def\e{\varepsilon}

\def\Oh{\mathcal{O}}

\def\det{\mathrm{det}}

\def\GL{\mathrm{GL}}
\def\SL{\mathrm{SL}}
\def\Hom{\mathrm{Hom}}
\def\Bil{\mathrm{Bil}}
\def\lan{\langle}
\def\ran{\rangle}
\def\ot{\otimes}
\def\Gc{G}
\def\Hc{H}
\def\Sc{K}
\def\Pc{\Pi}

\def\ti{\times}
\def\ul{\underline}
\def\mamU{\mathrm{M}_{\ul{U}}}
\def\mult{\mathrm{mult}}
\def\ula{{\ul{\lambda}}}
\def\umu{{\ul{\mu}}}

\def\ua{{\ul{\alpha}}}
\def\um{{\ul{m}}}

\def\diag{\mathrm{diag}}
\def\sgn{\mathrm{sgn}}
\def\mamun{{\lan n,n,n\ran}}
\def\br{{\underline{R}}}

\newcommand{\ol}[1]{{\overline{#1}}}

\def\res{{\mathrm{res}}}
\def\Id{{\mathrm{id}}}
\def\Sym{{\mathsf{Sym}}}

\def\stab{{\mathrm{stab}}}
\def\Par{{\mathsf{Par}}}
\def\meet{\curlywedge}

\def\im{\mathrm{im}}

\def\Tc{D}

\def\Aut{\mathrm{Aut}}
\def\Bil{\mathrm{Bil}}
\def\End{\mathrm{End}}
\def\ux{\underline{x}}

\newcommand{\reg}{\ensuremath{\mathsf{reg}}}
\newcommand{\len}{\mathsf{len}}

\def\HWV{{\mathcal H}}
\def\span{\mathrm{span}}
\def\Tau{\mathcal{T}}
\def\Sp{\mathcal{S}}
\def\Perm{\mathrm{Perm}}
\def\St{\mathrm{St}}